%% file: main.tex
\documentclass[5p,times]{elsarticle}

\usepackage[numbers]{natbib}
\usepackage{graphicx}
\usepackage{url}
\usepackage{xcolor}
\usepackage{stfloats}
\usepackage{tcolorbox}
\usepackage{amsmath}
\usepackage{amsfonts}
\usepackage{multirow}
\usepackage{booktabs}
\usepackage{listings}
\usepackage{pifont}
\usepackage{lineno}
\usepackage{hhline}
\usepackage{diagbox}        
\usepackage{adjustbox}      
\usepackage{subcaption} 
\usepackage{amsthm}
\newtheorem{proposition}{Proposition}
\newtheorem{remark}{Remark}

\definecolor{verylightgray}{rgb}{.97,.97,.97}

\newcommand{\syh}[1]{\textcolor{black}{#1}}
\newcommand{\yg}[1]{\textcolor{black}{#1}}
\newcommand{\tool}{\texttt{VeriSelector}}

\journal{Information Sciences}

\begin{document}

\begin{frontmatter}
	
\title{Quality over Quantity: Diversity-Aware Data Selection for Efficient Verilog Code Generation}

\author[NPU,NTNC]{Yiheng Shen}
\ead{yiheng.s@ntnc.edu.cn}

\author[NPU]{Wei Zheng\corref{first_corresponding_author}}
\cortext[first_corresponding_author]{Corresponding author}
\ead{wzheng@nwpu.edu.cn}

\author[NPU]{Xiao Wei}
\ead{wxgd@mail.nwpu.edu.cn}

\author[ECNU]{Hao Shen}
\ead{52285902020@stu.ecnu.edu.cn}

\author[NTU]{Xiang Chen}
\ead{xchencs@ntu.edu.cn}

\author[ZJU]{Guang Yang\corref{second_corresponding_author}}
\cortext[second_corresponding_author]{Corresponding author}
\ead{guang_yang@zju.edu.cn}

\address[NPU]{Northwestern Polytechnical University, Xi'an, China}
\address[NTNC]{Nantong Normal College, Nantong, China}
\address[ECNU]{East China Normal University, Shanghai, China}
\address[NTU]{Nantong University, Nantong, China}
\address[ZJU]{Zhejiang University, Hangzhou, China}

\begin{abstract}
Large Language Models (LLMs) have shown remarkable potential in Verilog code generation, yet existing datasets contain considerable noise and redundancy. 
\syh{Prior data selection methods address only isolated quality aspects, neglect the global diversity of the training set, and cannot capture Verilog-specific structural semantics.}
To bridge this gap, we propose {\tool}, the first data selection framework for Verilog code generation that jointly optimizes quality and diversity. We formulate the selection problem as a constrained bi-objective subset selection problem and solve it via a three-stage approximation. For quality, a multi-granularity pipeline first verifies functional correctness through testbench simulation and then filters misaligned samples via Instruction-Following Difficulty (IFD) scoring. For diversity, 109-dimensional Verilog-specific structural features (AST, CFG, and Netlist) are fused with textual embeddings for clustering-based diversity modeling. A proportional adaptive sampling strategy then allocates per-cluster quotas guided by IFD ranks, with a provable distribution preservation guarantee. Experiments on three LLMs and three benchmarks show that {\tool} outperforms full-dataset training and state-of-the-art baselines using only 20\%--25\% of the data, achieving Performance Retention Rates above 118\% and reducing training time by over 80\%. Notably, {\tool} improves average Pass@1 by 18.49\%--29.43\% over full-dataset training and by 1.36\%--7.95\% over the best-performing baseline across all evaluated models.
\end{abstract}


\end{frontmatter}

\begin{sloppypar}

\input{sections/1.intro}
\input{sections/2.preliminaries}
\input{sections/3.method}
\input{sections/4.setup}
\input{sections/5.results}
\input{sections/6.discussion}
\input{sections/7.related}
\input{sections/8.conclusion}

\section*{Acknowledgement}
The authors would like to thank the editors and the anonymous reviewers for their insightful comments and suggestions, \syh{which have substantially improved the quality of this work.}
Guang Yang is supported by the Postdoctoral Fellowship Program of CPSF under Grant Number GZC20260902.

	


\bibliography{reference}
\bibliographystyle{elsarticle}
\end{sloppypar}
\end{document}

%% file: sections/1.intro.tex
\section{Introduction}
\label{sec:intro}

Recent advances in Large Language Models (LLMs)~\cite{yang2026less} for general code generation~\cite{yang2025code} have motivated researchers to apply them to Verilog code generation~\cite{liu2024rtlcoder}.
\syh{Verilog is the predominant Hardware Description Language (HDL) for chip design~\cite{thakur2024verigen}. The efficiency and quality of Verilog implementations strongly affect the chip development cycle and engineering reliability.}
Thus, enabling LLMs to generate correct and high-quality Verilog code is of considerable practical value.

For Verilog code generation, a critical challenge is the lack of high-quality, well-curated domain-specific datasets.
To alleviate this issue, \syh{existing efforts~\cite{liu2024rtlcoder} construct Verilog datasets by crawling open-source GitHub repositories or by synthesizing samples with LLMs~\cite{li2023starcoder}.}
However, such construction methods inevitably introduce substantial noise and information redundancy~\cite{kandpal2022deduplicating}, severely impairing model training efficiency and generalization capability.

A natural solution to this problem is data filtering and selection.
Several preliminary efforts have been made in this direction: Liu et al.~\cite{liu2024rtlcoder} introduce syntax-level filtering to eliminate erroneous samples, and Wei et al.~\cite{wei2025vericoder} employ LLM-generated test cases with compilation-simulation verification to validate functional correctness.
While these efforts represent meaningful progress, they address only isolated aspects of data quality and lack a systematic data selection framework that jointly considers quality, diversity, and domain-specific characteristics. 
To illustrate, our empirical analysis reveals that training on the full RTLCoder-27k dataset yields average Pass@1 scores substantially lower than training on a carefully selected 20\%--25\% subset (e.g., 37.20\% vs.\ 44.08\% for Qwen2.5-Coder 7B), confirming that indiscriminate data usage introduces noise that harms model performance.
Specifically, when applied to the Verilog domain, existing methods still exhibit two key limitations.



\textbf{Limitation 1: Lack of domain-aware diversity modeling.}
Current data selection approaches~\cite{liu2024rtlcoder} focus predominantly on per-sample quality assessment without considering the global distribution of the training set~\cite{wei2025vericoder}.
As a consequence, even when individual samples satisfy quality criteria, the resulting dataset tends to be concentrated on a narrow range of hardware design scenarios or module types, leading to overall homogenization that restricts model generalization.
Addressing this diversity gap requires capturing the structural semantics specific to hardware designs.
Unlike general-purpose programming languages, Verilog embeds domain-specific characteristics~\cite{loow2025simulation} such as timing logic and parallel execution, and its functional description relies on hardware-specific representations, including Netlists~\cite{wolf2013yosys}.
However, existing general-purpose diversity-oriented selection methods~\cite{yu2024diversify} rely solely on textual semantic embeddings for clustering and sampling, and cannot capture these structural and behavioral properties~\cite{renduchintala2024smart}.
Meanwhile, existing Verilog-specific methods~\cite{liu2024rtlcoder} remain limited to syntax-level or compilation-level checking and do not model the deeper structural semantics of hardware designs either~\cite{wei2025vericoder}.
As a result, there is currently no method capable of performing diversity-aware data selection grounded in Verilog's domain-specific structural features.

\textbf{Limitation 2: Lack of a unified framework.}
Beyond the individual shortcomings described above, no existing framework integrates multi-granularity quality selection and domain-aware diversity optimization into a single end-to-end pipeline.
Quality-oriented methods~\cite{liu2024rtlcoder} perform one-off syntax or functional verification in isolation, easily leading to redundant sampling of high-quality but homogeneous samples~\cite{wei2025vericoder}.
Diversity-oriented methods~\cite{yu2024diversify}, on the other hand, optimize distributional coverage without enforcing quality constraints, risking the inclusion of noisy or incorrect samples~\cite{renduchintala2024smart}.
Existing studies thus either discard diversity requirements for strict quality control, or sacrifice sample quality for distributional coverage, resulting in a fundamental inability to balance the two core objectives.
This fragmentation severely limits the efficiency and effectiveness of LLM fine-tuning for hardware design tasks.

To address these limitations, we propose {\tool}, the first systematic quality-and-diversity-aware data selection framework for Verilog code generation.
To ensure \textbf{data quality}, {\tool} employs a multi-granularity selection pipeline: it first verifies functional correctness through automated testbench generation and simulation~\cite{wei2025vericoder}, and then assesses instruction-solution alignment via Instruction-Following Difficulty (IFD) scoring~\cite{li2024quantity}, progressively eliminating noisy samples from behavioral to semantic levels.
To ensure \textbf{data diversity}, {\tool} introduces a Verilog-specific multi-level structural representation, including AST, CFG, and Netlist features with a total of 109 dimensions (57 dimensions for AST, 28 dimensions for CFG, and 24 dimensions for Netlist). Then {\tool} fuses the structural features with textual semantic embeddings and performs clustering-based diversity modeling. 
Building on this, a proportional \textbf{adaptive sampling} strategy allocates per-cluster quotas guided by IFD ranks, selecting a compact subset that jointly preserves structural diversity and prioritizes high-quality samples.

\syh{To evaluate {\tool}, we conducted experiments on the RTLCoder-27k dataset~\cite{liu2024rtlcoder}.
This corpus is constructed by expanding keywords and generating code variants with GPT-4o.}
\syh{We evaluated three models of similar scale (Qwen2.5-Coder 7B, SeedCoder 8B, and CodeLlama 7B) on three benchmarks (VerilogEval-v2~\cite{liu2023verilogeval}, ResBench~\cite{guo2025resbench}, and RTLLM~\cite{lu2024rtllm}).}
Using only 20\% to 25\% of the training data, {\tool} surpasses full-dataset training and state-of-the-art baselines, achieving Performance Retention Rates over 118\% across all three models.
Moreover, ablation studies confirm the contribution of each component, while t-SNE visualization demonstrates that the selected subset retains the distributional diversity of the full dataset.

In summary, this paper makes the following contributions:

\begin{itemize}
    \item First Verilog data selection framework with provable quality-diversity bounds.
    \item A multi-granularity quality pipeline from correctness to instruction alignment.
    \item AST/CFG/Netlist features fused with text for diversity-aware clustering.
\end{itemize}

To facilitate future reproduction and follow-up work, we release source code, models, and datasets at \url{https://github.com/syhstudy/VeriSelecter} (accessed 5 September 2026).

%% file: sections/2.preliminaries.tex
\section{Preliminaries}
\label{sec:preliminaries}

In this section, we introduce the background of LLM-driven Verilog code generation, the structural features of Verilog code, and formally define the dataset selection problem addressed in this work.

\subsection{LLM-Driven Verilog Code Generation}
Given a natural language specification $x_i$ describing a hardware design requirement, the Verilog code generation model $M$ aims to autoregressively generate the corresponding Verilog code $y_i$. This process is parameterized by $\theta$ and can be expressed as:
\begin{equation}
P_{\theta}(y_i|x_i)=\prod_{k=1}^{n}P_{\theta}(y_{i,k}|x_i,y_{i,1:k-1}) 
\end{equation}
where $y_{i,1:k-1}$ represents the previously generated tokens before the $k$-th token of $y_i$, and $n$ denotes the total number of tokens in the target sequence $y_i$. 

To adapt LLMs to the Verilog code generation task, we employ Supervised Fine-Tuning (SFT)~\cite{ouyang2022training} combined with Low-Rank Adaptation (LoRA)~\cite{hu2022lora}. SFT minimizes the negative log-likelihood loss over a training dataset $\mathcal{D}$:
\begin{equation}
\mathcal{L}_{\text{SFT}} = E_{(x,y) \sim \mathcal{D}} \left[ -\log P_{\theta}(y \mid x) \right]
\label{eq:sft_loss}
\end{equation}
LoRA freezes the original weight matrix $W \in \mathbb{R}^{d \times k}$ and introduces a low-rank update $\Delta W = \frac{\alpha}{r} \cdot BA$, where $B \in \mathbb{R}^{d \times r}$ and $A \in \mathbb{R}^{r \times k}$ are trainable low-rank matrices with $r \ll \min(d, k)$, $\alpha$ is the scaling factor, and $r$ is the rank. Only $A$ and $B$ are optimized during training. Further details of the SFT and LoRA configurations are provided in Section~\ref{sec:setup}.

\subsection{Structural Features of Verilog Code}
Unlike general-purpose programming languages, Verilog code can be characterized from multiple structural levels, including AST, CFG, and Netlist, which simultaneously reflect its software semantics and hardware design nature.
These structural perspectives can be quantified into numerical feature vectors that characterize different aspects of Verilog code, providing a foundation for the data selection problem defined below.

\begin{figure}[t]
    \centering
    \includegraphics[width=\linewidth]{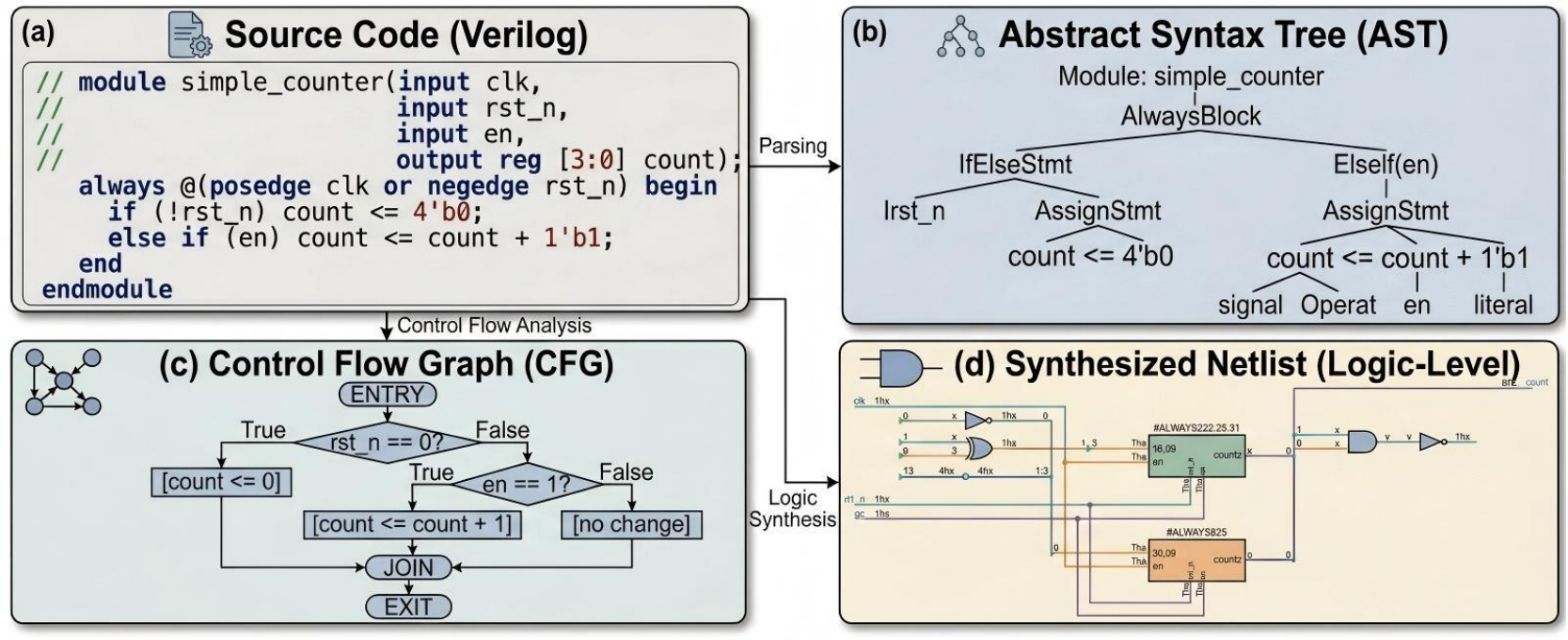}
    \caption{An illustrative Verilog module and its three structural representations: (a) source code, (b) Abstract Syntax Tree (AST), (c) Control Flow Graph (CFG) of the \texttt{always} block, and (d) synthesized netlist.}
    \label{fig:verilog_features}
\end{figure}

As shown in Figure~\ref{fig:verilog_features}, a Verilog module can be analyzed from three complementary structural perspectives, corresponding to panels (b)--(d):

\textbf{AST Features.} Abstract Syntax Tree (AST)~\cite{alon2018general} features capture the static syntactic structure of code. By analyzing node types in the tree (e.g., module declarations, continuous assignments, procedural blocks, operators, operands, etc.), tree depth, width, and child node distribution, AST features can quantify the syntactic complexity and structural normativity of code, facilitating the identification of redundant or anomalous syntactic patterns.
    
\textbf{CFG Features.} Control Flow Graph (CFG)~\cite{allen1970control} features characterize the dynamic execution logic of behavioral-level code (e.g., always blocks), where nodes correspond to basic blocks (sets of sequentially executed statements) and edges represent control flow transfers between blocks (e.g., conditional branches, loops). Analyzing the scale, path complexity, loop structures, and reachability of CFGs is critical for verifying the completeness of code logic and detecting unreachable code and logical conflicts.
    
\textbf{Netlist Features.} Netlist features~\cite{wang2022functionality} directly describe the physical and logical structures of hardware, including hierarchical relationships of modules (e.g., instantiation connections between top-level modules and sub-modules), port directions and types (input, output, inout), combinational and sequential logic connections between signals, and timing constraints (e.g., setup/hold time). Such information serves as the fundamental basis for evaluating the functional correctness, timing compliance, and structural integrity of designs.


\subsection{Problem Formulation}
Based on the structural features introduced above, we now formally define the quality-diversity-aware dataset selection problem for Verilog code generation.

Let $\mathcal{D} = \{d_1, d_2, \dots, d_N\}$ denote the original Verilog training dataset, where each sample $d_i = \langle \text{Instruction}_i, \text{Code}_i \rangle$ consists of a natural language instruction and its canonical Verilog solution. For each sample, we extract a textual feature vector $\mathbf{f}_{\text{Text}}(d_i)$ encoding the semantic information of the instruction, and a multi-level structural feature vector:
\[
\mathbf{f}_{\text{Struct}}(d_i) = [\mathbf{f}_{\text{AST}}(d_i),\; \mathbf{f}_{\text{CFG}}(d_i),\; \mathbf{f}_{\text{Netlist}}(d_i)]
\]
These are fused into a unified representation $\mathbf{e}(d_i)$ that captures both the semantic intent and the structural characteristics of each sample.

The objective is to design a selection function
\[
\mathcal{S}: \mathcal{D} \to \mathcal{D}', \quad \text{where} \quad \mathcal{S}(\mathcal{D}) = \mathcal{D}' \subseteq \mathcal{D}, \quad |\mathcal{D}'| \ll |\mathcal{D}|
\]
guided by the fused representations $\{\mathbf{e}(d_i)\}_{i=1}^{N}$, such that the selected subset $\mathcal{D}'$ simultaneously satisfies the following constraints:
\begin{itemize}
    \item \textbf{Quality constraint}: Samples in $\mathcal{D}'$ shall exhibit high functional correctness and strong instruction-solution alignment;
    \item \textbf{Diversity constraint}: $\mathcal{D}'$ shall cover diverse hardware design scenarios (e.g., combinational logic, sequential logic) and module types (e.g., adders, registers, etc.);
    \item \textbf{Efficiency constraint}: The selection process shall have low computational cost, and the performance (e.g., Pass@k) of models trained on $\mathcal{D}'$ shall be comparable to that of models trained on the full dataset $\mathcal{D}$.
\end{itemize}


\subsection{Theoretical Analysis}
\label{sec:theory}

Building on the problem formulation above, we formalize the selection objective that {\tool} approximates and establish a provable guarantee on its diversity preservation.

\subsubsection{Unified Optimization Formulation.}
The quality-aware selection stage constructs a feasible set $\mathcal{D}_{\text{quality}}$ by applying hard constraints. Let $Q_i = \text{Instruction}_i$ and $A_i = \text{Code}_i$ denote the instruction and code of sample $d_i$, respectively. The feasible set is defined as:
\begin{equation}
\mathcal{D}_{\text{quality}} = \{ d_i \in \mathcal{D} \mid \text{Sim}(A_i,\, t_i) = \texttt{True} \;\wedge\; \text{IFD}(Q_i,\, A_i) < 1 \}
\label{eq:feasible}
\end{equation}
This set excludes all functionally incorrect samples and all instruction-solution pairs that lack semantic alignment (see Section~\ref{sec:quality} for details).

Given the K-Means partition $\{\mathcal{C}_k\}_{k=1}^K$ of $\mathcal{D}_{\text{quality}}$ in the fused feature space and a sampling budget $N_{\text{target}}$, define the cluster distribution $\hat{\pi}_k = |\mathcal{C}_k| / |\mathcal{D}_{\text{quality}}|$ and the selected distribution $\hat{\pi}'_k = n_k / N_{\text{target}}$. The selection objective can be written as:
\begin{equation}
\max_{\mathcal{D}' \subseteq \mathcal{D}_{\text{quality}},\; |\mathcal{D}'| \le N_{\text{target}}}
  \sum_{d_i \in \mathcal{D}'} \psi(d_i)
  \quad \text{s.t.} \quad
  \|\hat{\pi}' - \hat{\pi}\|_1 \le \epsilon
\label{eq:objective}
\end{equation}
where $\psi(d_i) = \text{IFD}(Q_i, A_i)$ serves as the learning signal that prioritizes moderately challenging yet well-aligned samples.

{\tool} approximates this objective via sequential decomposition. Quality filtering constructs $\mathcal{D}_{\text{quality}}$. Proportional quota allocation enforces the distribution constraint. Intra-cluster IFD ranking maximizes $\sum \psi$. The three stages collectively address the quality, diversity, and efficiency constraints defined in Section~\ref{sec:preliminaries}.

\subsubsection{Distribution Preservation Guarantee.}
We now show that the proportional sampling strategy preserves the cluster distribution of $\mathcal{D}_{\text{quality}}$ with a bounded deviation.

\begin{proposition}[Distribution Preservation]
\label{prop:dist}
Let $\mathcal{D}_{\text{quality}}$ be partitioned into $K$ clusters, and let the per-cluster quota be $n_k = \mathrm{round}(|\mathcal{C}_k| / |\mathcal{D}_{\text{quality}}| \cdot N_{\text{target}})$. Then the selected subset satisfies:
\begin{equation}
\|\hat{\pi}' - \hat{\pi}\|_1 \le \frac{2K}{N_{\text{target}}}
\label{eq:dist_bound}
\end{equation}
\end{proposition}

\begin{proof}
The ideal continuous quota $\tilde{n}_k = |\mathcal{C}_k| \cdot N_{\text{target}} / |\mathcal{D}_{\text{quality}}|$ satisfies $\tilde{n}_k / N_{\text{target}} = \hat{\pi}_k$ exactly. The rounding operation introduces a per-cluster error $|n_k - \tilde{n}_k| \le 1$, which translates into a distributional error of $|\hat{\pi}'_k - \hat{\pi}_k| \le 1 / N_{\text{target}}$ per cluster. Summing over all $K$ clusters yields $\|\hat{\pi}' - \hat{\pi}\|_1 \le K / N_{\text{target}}$. An additional correction of at most $K / N_{\text{target}}$ may be needed to enforce the exact budget $\sum_k n_k = N_{\text{target}}$. This gives the stated bound $2K / N_{\text{target}}$.
\end{proof}

In our setting, $K = 100$ and $N_{\text{target}} \approx 5{,}500$ (20\% of the original dataset). The bound evaluates to approximately 0.036, indicating less than 4\% distributional deviation in $L_1$ norm. This provides a formal guarantee that proportional sampling preserves the structural diversity of $\mathcal{D}_{\text{quality}}$.

\begin{remark}[Parameter-Free IFD Threshold]
\label{rmk:ifd}
By definition, $\mathrm{IFD}(Q, A) < 1$ holds if and only if $\mathcal{L}_\theta(A \mid Q) < \mathcal{L}_\theta(A)$. That is, the instruction strictly reduces the generation loss of the solution under the prior model. The threshold $\tau = 1$ is derived directly from this definition and requires no validation-set tuning.
\end{remark}

%% file: sections/3.method.tex
\section{Methodology}
\label{sec:method}

\subsection{Overview}
As illustrated in Figure~\ref{fig:framework}, {\tool} consists of three stages: (1) \textbf{Quality-Aware Selection}, which progressively eliminates noisy samples via testbench verification and IFD-based alignment scoring; (2) \textbf{Diversity-Aware Selection}, which fuses textual semantics with Verilog-specific structural features (AST, CFG, and Netlist) to partition samples into clusters capturing distinct hardware design patterns; and (3) \textbf{Adaptive Sampling}, which applies intra-cluster quality-ranked proportional sampling to produce the final training dataset $\mathcal{D}'$.
Notably, the core novelty lies in the Verilog-specific structural knowledge representation and its coupling with quality signals, rather than in the clustering algorithm itself.

\begin{figure*}[t!]
    \centering
    \includegraphics[width=\textwidth]{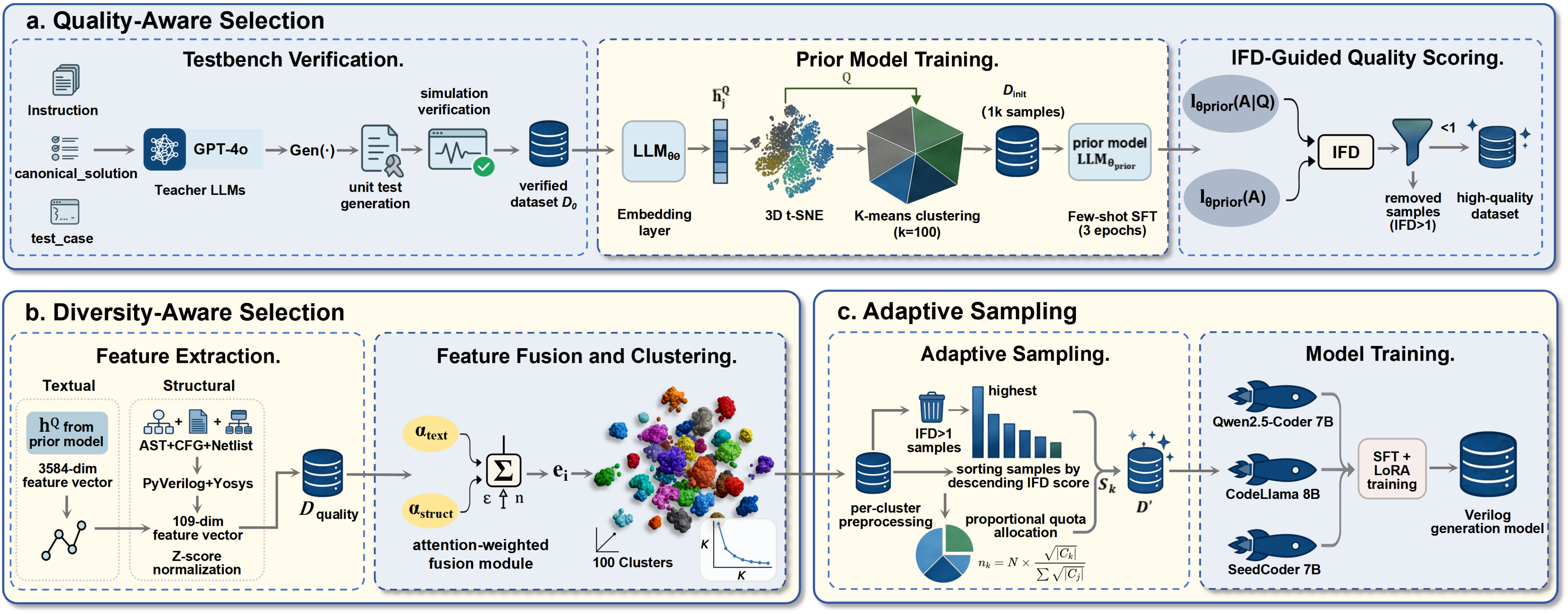}
    \caption{Framework and workflow of {\tool}.}
    \label{fig:framework}
\end{figure*}

\subsection{Quality-Aware Selection}
\label{sec:quality}
The quality-aware selection stage aims to remove samples that are functionally incorrect or exhibit misalignment between instructions and solutions. It consists of three steps: testbench verification, prior model training, and IFD-guided scoring.
To ensure data quality, we design a multi-granularity quality selection pipeline at this stage.

\subsubsection{Testbench Verification.}
Given the original dataset $\mathcal{D} = \{d_i\}_{i=1}^N$ where each sample $d_i = \langle \text{Instruction}_i, \text{Code}_i \rangle$, the first step eliminates functionally incorrect samples through automated testbench generation and simulation.

For each sample $d_i$, we employ GPT-4o to generate a testbench based on its instruction and code~\cite{wei2025vericoder}:
\begin{equation}
t_i = \text{Gen}(\text{GPT-4o},\; \text{Instruction}_i,\; \text{Code}_i)
\end{equation}

Each generated testbench $t_i$ is then executed against $\text{Code}_i$ using a standardized Verilog simulation toolchain. A sample is retained only if its code passes the compilation and simulation verification:
\begin{equation}
\mathcal{D}_0 = \{d_i \in \mathcal{D} \mid \text{Sim}(\text{Code}_i,\; t_i) = \texttt{True}\}
\end{equation}

\subsubsection{Prior Model Training.}
\label{sec:prior_model}
To enable efficient IFD scoring without relying on LLMs, we train a lightweight prior model on a small diverse subset. 
This design is motivated by the finding that IFD scores from smaller models are highly consistent with those from stronger models~\cite{li2024superfiltering}, making it possible to perform reliable quality assessment at low computational cost.

Starting from the base model $\text{LLM}_{\theta_0}$, we extract instruction embeddings for all samples in $\mathcal{D}_0$ by mean-pooling the last hidden states:
\begin{equation}
\boldsymbol{h}_i^{Q(0)} = \frac{1}{m}\sum_{j=1}^m \boldsymbol{h}_{i,j}^Q
\end{equation}
where $[\boldsymbol{h}_{i,1}^Q, \dots, \boldsymbol{h}_{i,m}^Q] = \text{LLM}_{\theta_0}(\text{Instruction}_i)$. We then apply K-Means clustering on $\{\boldsymbol{h}_i^{Q(0)}\}$ to partition $\mathcal{D}_0$ into 100 clusters and randomly sample 10 samples per cluster, yielding a 1{,}000-sample initial set $\mathcal{D}_{\text{init}}$. Inspired by LIMA~\cite{zhou2023lima}, we fine-tune $\text{LLM}_{\theta_0}$ on $\mathcal{D}_{\text{init}}$ for 3 epochs with autoregressive cross-entropy loss, obtaining the prior model $\text{LLM}_{\theta_{\text{prior}}}$.

Note that the clustering here serves solely to ensure diversity in the prior model's training data. It is distinct from the diversity-aware clustering in Section~\ref{sec:diversity}, which operates on richer fused features for the final selection.

\subsubsection{IFD-Guided Quality Scoring.}
\label{sec:ifd}
After functional correctness filtering, we further assess the semantic alignment between each instruction-solution pair. Specifically, we adopt the Instruction-Following Difficulty (IFD) metric~\cite{li2024quantity} computed by the prior model $\text{LLM}_{\theta_{\text{prior}}}$.
Let $Q$ and $A$ denote the instruction and solution of a given sample (i.e., $Q = \text{Instruction}_i$ and $A = \text{Code}_i$), $w_k^A$ the $k$-th token of $A$, and $N_A$ the total number of tokens in $A$. The IFD score is defined as the ratio of two conditional losses:

\textbf{With instruction context:}
\begin{equation}
\mathcal{L}_{\theta_{\text{prior}}}(A|Q) = -\frac{1}{N_A}\sum_{k=1}^{N_A}\log P(w_k^A \mid Q, w_1^A, \dots, w_{k-1}^A;\; \theta_{\text{prior}})
\end{equation}

\textbf{Without instruction context:}
\begin{equation}
\mathcal{L}_{\theta_{\text{prior}}}(A) = -\frac{1}{N_A}\sum_{k=1}^{N_A}\log P(w_k^A \mid w_1^A, \dots, w_{k-1}^A;\; \theta_{\text{prior}})
\end{equation}

The IFD score is then computed as:
\begin{equation}
\text{IFD}(Q, A) = \frac{\mathcal{L}_{\theta_{\text{prior}}}(A|Q)}{\mathcal{L}_{\theta_{\text{prior}}}(A)}
\end{equation}

Intuitively, $\text{IFD} = 1$ constitutes a natural decision boundary derived from the definition of the metric itself~\cite{li2024quantity}: when $\text{IFD} < 1$, the conditional loss $\mathcal{L}(A|Q) < \mathcal{L}(A)$, indicating that the instruction reduces the difficulty of generating the solution and thus provides positive guidance; when $\text{IFD} = 1$, the instruction contributes no additional information beyond what the model can infer from the solution alone; when $\text{IFD} > 1$, the instruction increases generation difficulty, suggesting semantic misalignment or even interference between the instruction and the solution.

Unlike thresholds that require empirical tuning, $\text{IFD} = 1$ is a \textit{parameter-free} boundary that naturally separates well-aligned from misaligned instruction--solution pairs. We therefore discard all samples with $\text{IFD} \geq 1$, obtaining the quality-filtered dataset:
\begin{equation}
\mathcal{D}_{\text{quality}} = \{d_i \in \mathcal{D}_0 \mid \text{IFD}(Q_i, A_i) < 1\}
\end{equation}

\subsection{Diversity-Aware Selection}
\label{sec:diversity}
After quality filtering, the remaining challenge is to ensure that the selected subset covers diverse hardware design scenarios. To this end, we extract dual-dimensional representations, namely textual semantics and structural features, for each sample, fuse them into a unified representation, and partition samples into clusters that capture distinct Verilog design patterns.
Thus, we innovatively incorporate Verilog-specific hardware features (especially Netlist features) to enhance the accuracy of clustering.

\subsubsection{Feature Extraction.}

\textbf{Textual Features.}
For each sample in $\mathcal{D}_{\text{quality}}$, we extract instruction embeddings using the prior model $\text{LLM}_{\theta_{\text{prior}}}$ by mean-pooling its last hidden states, yielding $\boldsymbol{h}_i^Q \in \mathbb{R}^{d_t}$. These embeddings encode richer semantic information than those from the base model, as the prior model has acquired basic Verilog instruction-following capability through fine-tuning.

\textbf{Structural Features.}
Unlike general-purpose code, Verilog code can be characterized at multiple structural levels that reflect both its software semantics and hardware design nature. 
We exploit this property by extracting features from three complementary perspectives: AST features capture syntactic design patterns (e.g., combinational vs.\ sequential logic constructs), CFG features characterize behavioral complexity (e.g., state machine depth and branching density), and Netlist features reflect physical implementation characteristics (e.g., module hierarchy and signal connectivity). 
Together, they provide a domain-specific structural fingerprint that textual embeddings alone cannot express.

Concretely, we employ PyVerilog to parse each $\text{Code}_i$ and extract AST and CFG features, and Yosys to synthesize the code and extract Netlist features. The resulting structural feature vector is:
\begin{equation}
\boldsymbol{f}_i^{\text{struct}} = [\boldsymbol{f}_i^{\text{AST}};\; \boldsymbol{f}_i^{\text{CFG}};\; \boldsymbol{f}_i^{\text{Netlist}}] \in \mathbb{R}^{d_s}
\end{equation}
where $d_s = 109$ (57 dimensions for AST, 28 dimensions for CFG, and 24 dimensions for Netlist) in our implementation.
Table~\ref{tab:feature_dimension} illustrates the composition of the feature dimensions, and additional details can be \syh{found on the project homepage~\footnote{\url{https://github.com/syhstudy/VeriSelecter/blob/main/feature_section.md} (accessed 5 September 2026)}.}

\begin{table}[htbp]
  \centering
  \caption{Feature Dimension Composition Details.}
  \label{tab:feature_dimension}
  \resizebox{0.45\textwidth}{!}{
  \begin{tabular}{ccl}
  \toprule
  \textbf{Feature} & \textbf{Dimension} & \textbf{Details} \\
  \midrule
  \multirow{2}{*}{AST} & \multirow{2}{*}{57} & Node type (54), depth parameters (1), \\
                      &                     & total node count (1), node count per depth (1) \\
  \multirow{2}{*}{CFG} & \multirow{2}{*}{28} & Basic features (10), label features (7), \\
                      &                     & kernel features (11) \\
  \multirow{2}{*}{Netlist} & \multirow{2}{*}{24} & Basic features (9), type features (6), \\
                           &                     & gate features (9) \\
  \bottomrule
  \end{tabular}}
\end{table}

\subsubsection{Feature Fusion and Clustering.}

\textbf{Standardization.}
To eliminate scale discrepancies, we apply Z-score normalization to textual features $\boldsymbol{h}_i^Q$ and structural features $\boldsymbol{f}_i^{\text{struct}}$ independently, yielding $\hat{\boldsymbol{h}}_i^Q$ and $\hat{\boldsymbol{f}}_i^{\text{struct}}$.

\textbf{Fusion.}
Since both feature modalities have been independently standardized to zero mean and unit variance via Z-score normalization, we directly concatenate them to form the fused representation:
\begin{equation}
\mathbf{e}_i = \hat{\boldsymbol{h}}_i^Q \;\oplus\; \hat{\boldsymbol{f}}_i^{\text{struct}}
\end{equation}
where $\oplus$ denotes vector concatenation. The prior Z-score normalization ensures that the two modalities contribute on a comparable scale, eliminating the need for additional fusion weights.

\textbf{Clustering.}
We apply K-Means on the fused representations $\{\mathbf{e}_i\}$ to partition $\mathcal{D}_{\text{quality}}$ into $K$ clusters (determined by the elbow method; $K=100$ in our experiments):
\begin{equation}
\min_{\{\mathcal{C}_k\}_{k=1}^K} \sum_{k=1}^K \sum_{i \in \mathcal{C}_k} \|\mathbf{e}_i - \boldsymbol{\mu}_k\|_2^2
\end{equation}
where $\boldsymbol{\mu}_k$ is the centroid of cluster $\mathcal{C}_k$. Each cluster captures a coherent set of Verilog design patterns, providing a structured foundation for diverse sampling.

\subsection{Adaptive Sampling}
\label{sec:sampling}
The final stage produces the selected dataset $\mathcal{D}'$ by performing quality-ranked proportional sampling within each cluster, achieving a balance between sample quality and scenario diversity.

\textbf{Intra-cluster Ranking.}
Within each cluster $\mathcal{C}_k$, samples are sorted in descending order by their IFD scores. The intuition is that samples with higher IFD scores (while still below 1) represent moderately challenging instruction-following tasks that provide stronger learning signals for the model~\cite{li2024quantity}, and are thus prioritized for selection.

\textbf{Proportional Sampling.}
Given a global sampling budget $N_{\text{target}}$, the quota for each cluster is allocated proportionally to its size:
\begin{equation}
n_k = \text{round}\left(\frac{|\mathcal{C}_k|}{|\mathcal{D}_{\text{quality}}|} \cdot N_{\text{target}}\right)
\end{equation}
where $|\mathcal{C}_k|$ is the number of samples in cluster $k$, and $|\mathcal{D}_{\text{quality}}|$ is the total number of quality-filtered samples. This ensures that the selected subset faithfully reflects the cluster distribution of $\mathcal{D}_{\text{quality}}$, preserving the natural proportion of different hardware design patterns.

The top-$n_k$ samples (by IFD ranking) from each cluster are selected, and the final dataset is assembled as:
\begin{equation}
\mathcal{D}' = \bigcup_{k=1}^K \text{Top}_{n_k}(\mathcal{C}_k)
\end{equation}
This dataset $\mathcal{D}'$ simultaneously satisfies the quality, diversity, and efficiency constraints defined in Section~\ref{sec:preliminaries}, and is used for subsequent model fine-tuning.

%% file: sections/4.setup.tex
\section{\syh{Experimental Setup}}
\label{sec:setup}

In this section, we present the research questions, experimental subjects, baselines, evaluation metrics, and implementation details.

\subsection{Research Questions}
To systematically evaluate the effectiveness, robustness, and design rationality of {\tool}, we propose the following three research questions (RQs).

RQ1: How \textbf{effective} is {\tool} compared to existing data selection methods in Verilog code generation tasks?

This RQ evaluates whether the joint quality-diversity selection strategy of {\tool} outperforms existing methods that either lack Verilog domain-specific features or fail to balance sample quality and data diversity.
We compare {\tool} against all baselines on three benchmarks using Pass@k metrics.

RQ2: How \textbf{sensitive} is {\tool} to the data selection rate?

This RQ investigates the performance of {\tool} under different data selection rates. 
We set multiple selection rates from 5\% to 25\% to analyze how performance changes under different data volume constraints, and to evaluate its adaptability in data-scarce scenarios.

\noindent\textbf{RQ3: What is the contribution of each component in {\tool}?}

This RQ validates the design rationality of {\tool} through ablation studies.
We systematically remove or replace core components, including testbench verification, IFD quality scoring, and structural feature fusion, to examine their individual contributions and synergistic effects.

\subsection{Experimental Subject}

\subsubsection{Datasets}

We adopt RTLCoder-27k~\cite{liu2024rtlcoder} as our base selection dataset. \syh{This dataset comprises 27,532 instruction-code pairs}, where each sample consists of a natural language description of a digital circuit design requirement and its corresponding Verilog implementation.
The dataset is generated via an automated pipeline and widely adopted for model fine-tuning: GPT-4o is leveraged to expand keywords and generate code variations, followed by syntax validation to filter invalid samples.

To evaluate the quality of the selected data, we adopt three widely used Verilog benchmarks:
\begin{itemize}
    \item \textbf{VerilogEval-v2}~\cite{liu2023verilogeval} consists of 156 problems with manually written specifications and standard testbenches;
    \item \textbf{RTLLM-v2}~\cite{lu2024rtllm} comprises 50 typical hardware module tasks, enabling evaluation of code correctness in real-world hardware design scenarios;
    \item \textbf{ResBench}~\cite{guo2025resbench} comprises 56 problems focusing on more complex and real-world-oriented hardware designs.    
\end{itemize}
All benchmarks provide golden testbenches for functional verification.

\subsubsection{Verilog code generation models}
To evaluate the cross-model generality of {\tool}, we select three mainstream open-source code models with distinct architectures and training paradigms:
\begin{itemize}
    \item \textbf{Qwen2.5-Coder 7B}~\cite{hui2024qwen2} is specifically optimized for code tasks with efficient inference capabilities.
    \item \textbf{CodeLlama 7B}~\cite{roziere2023code} is a classic general-purpose large code model serving as a representative baseline.
    \item \textbf{SeedCoder 8B}~\cite{seed2025seed} adopts a model-centric data pipeline that leverages LLMs to curate code pretraining data, minimizing reliance on hand-crafted filtering rules.
\end{itemize}
\yg{These models represent the three most widely adopted open-source model families for Verilog code generation according to a recent systematic literature review~\cite{yang2025large}.} These models are complementary in pretraining objectives and domain characteristics~\cite{wei2025vericoder, gao2024autovcoder, yubeaton2025verithoughts}, enabling rigorous verification of the generality of {\tool}.

\subsection{Baselines}
\label{sec:baseline}
We compare {\tool} against four state-of-the-art baselines spanning three categories: no-strategy sampling, quality-oriented selection, and domain-specific filtering.

\textbf{Random} constructs training subsets via uniform sampling from the original dataset without introducing any selection bias, serving as the lower-bound baseline.

\textbf{EL2N}~\cite{paul2021deep} is a loss-guided filtering method that uses gradient norms in early training stages to identify hard samples and prioritizes high-loss instances to accelerate model convergence.

\textbf{Cherry}~\cite{li2024quantity} is a self-guided data evaluation framework that enables models to automatically identify high-quality samples from large-scale datasets based on instruction-following difficulty.

\textbf{VeriCoder}~\cite{wei2025vericoder} is the only Verilog-specific baseline. It generates unit test cases and performs iterative simulation-based verification to remove functionally incorrect samples. Unlike {\tool}, VeriCoder focuses solely on functional correctness filtering and does not consider instruction-solution alignment quality or structural diversity.

\subsection{Evaluation Metrics}

\textbf{Pass@k.} We adopt the standard Pass@k metric~\cite{chen2021evaluating} as our primary evaluation metric. It assesses the probability that at least one out of $k$ generated candidates passes all unit tests, measuring the functional correctness of generated code. We report Pass@1, Pass@5, and Pass@10 in our experiments.

\textbf{Data Selection Rate (DSR)}. The DSR metric reflects the data utilization rate, defined as:
\begin{equation}
\nonumber
DSR = |\mathcal{D}'| / |\mathcal{D}|,
\end{equation}
where $|\mathcal{D}'|$ denotes the size of the selected subset and $|\mathcal{D}|$ denotes the size of the original dataset.

\textbf{Performance Retention Rate (PRR)}. The PRR metric quantifies the degree of retention of key information, defined as:
\begin{equation}
\nonumber
PRR = \mathcal{P}(\mathcal{D}') / \mathcal{P}(\mathcal{D}),
\end{equation}
where $\mathcal{P}(\mathcal{D}')$ is the model performance trained on the selected subset, and $\mathcal{P}(\mathcal{D})$ is the model performance trained on the full dataset.

\subsection{Implementation Details}



\begin{table}[t]
\centering
\caption{Hyper-parameters and their values.}
\label{tab:hyperparams}
\resizebox{0.5\textwidth}{!}{
\begin{tabular}{c|c||c|c}
\toprule
\textbf{Hyper-Parameter} & \textbf{Value}
& \textbf{Hyper-Parameter} & \textbf{Value} \\
\midrule
Cluster Numbers $K$ & 100
& Selection Rate & 10\%--25\% \\
Training Epochs & 3
& Learning Rate & 2e-4 \\
\bottomrule
\end{tabular}}
\end{table}

Table~\ref{tab:hyperparams} summarizes the key hyper-parameters used in our experiments. 
The number of clusters $K$ is determined by the elbow method.
For inference, we adopt nucleus sampling with temperature $\tau = 0.8$ and top-$p = 0.95$. For each test problem, we generate 20 candidate solutions to compute unbiased estimates of Pass@1, Pass@5, and Pass@10.
We emphasize that all models are trained for a fixed number of epochs (3 epochs) with fixed hyper-parameters, rather than being trained until a target accuracy is reached. This design prevents overfitting to the training set. Moreover, the three evaluation benchmarks (VerilogEval-v2, RTLLM, and ResBench) are entirely independent of the training dataset RTLCoder-27k, ensuring that no test data leaks into the training or hyper-parameter tuning process.
All experiments are conducted on a platform equipped with an NVIDIA RTX 3090 (24GB) GPU, an Intel Xeon E5-2666 (2.9GHz) CPU, and 128GB RAM. The software environment is built upon Python 3.9 and PyTorch 2.1. We utilize Vivado 2023.1 as the Verilog compilation and simulation tool, Scikit-learn for clustering, Hugging Face Transformers for model loading and fine-tuning, and PyVerilog for code structure parsing.

%% file: sections/5.results.tex
\section{Result Analysis}
\label{sec:result}

\subsection{RQ1: Effectiveness Comparison}

To evaluate the performance of {\tool}, we compare it with state-of-the-art baselines discussed in Section~\ref{sec:baseline}.
For a fair comparison, we use the most efficient data utilization for each method and apply the same evaluation metrics.
The results are shown in Table~\ref{tab:RQ1}. 
Note that VeriCoder applies only testbench-based functional correctness filtering, making it a natural reference point for evaluating the additional value of {\tool}'s IFD scoring, structural feature fusion, and diversity-aware clustering components.

\begin{table*}[htbp]
\centering
\caption{The results of {\tool} on different models. The $\uparrow$/$\downarrow$ values denote absolute improvements/declines (in percentage points) over the All baseline. \textbf{Bold} indicates the best result among methods with the same or lower data selection rate. \yg{VeriCoder$_R$ denotes random subsampling from VeriCoder's quality-filtered pool to match the target DSR, serving as a controlled baseline that isolates the contribution of {\tool}'s IFD scoring, structural features, and diversity-aware clustering beyond testbench filtering.}}
\label{tab:RQ1}
\resizebox{\textwidth}{!}{
\begin{tabular}{ccccccccc}
\multicolumn{9}{c}{\rule{0pt}{1.5em}\textit{The Results of {\tool} on \syh{Qwen2.5-Coder 7B}}} \\
    \midrule
\multirow{2}{*}{\textbf{Benchmark}} & \multirow{2}{*}{\textbf{Pass@k}} & \textbf{All} & \textbf{VeriCoder} & \yg{\textbf{VeriCoder$_R$}} & \textbf{Random} & \textbf{EL2N} & \textbf{Cherry} & \textbf{Ours} \\
     & & \textbf{100.00\%} & \textbf{46.55\%} & \yg{\textbf{20.00\%}} & \textbf{20.00\%} & \textbf{20.00\%} & \textbf{20.00\%} & \textbf{20.00\%} \\
    \midrule
    \multirow{3}{*}{ResBench} & Pass@1 & 41.07\% & 55.36\% & \yg{53.57\%} & 51.79\% & 57.14\% & 55.36\% & \textbf{{57.14\%}($\uparrow 16.07\%$)} \\
    & Pass@5 & 62.50\% & 69.64\% & \yg{66.07\%} & 64.29\% & 64.69\% & 66.07\% & \textbf{{69.64\%}($\uparrow 7.14\%$)} \\
    & Pass@10 & 69.64\% & 69.64\% & \yg{69.64\%} & 69.64\% & 68.57\% & 71.43\% & \textbf{{71.43\%}($\uparrow 1.79\%$)} \\
    \midrule
    \multirow{3}{*}{RTLLM} & Pass@1 & 34.00\% & 36.00\% & \yg{34.00\%} & 34.00\% & 36.00\% & 36.00\% & \textbf{{36.00\%}($\uparrow 2.00\%$)} \\
    & Pass@5 & 40.00\% & 44.00\% & \yg{42.00\%} & 40.00\% & 38.00\% & \textbf{{48.00\%}} & 46.00\% ($\uparrow 6.00\%$)\\
    & Pass@10 & 40.00\% & 46.00\% & \yg{44.00\%} & 44.00\% & 42.00\% & 48.00\% & \textbf{{50.00\%}($\uparrow 10.00\%$)} \\
    \midrule
    \multirow{3}{*}{VerilogEval} & Pass@1 & 36.54\% & \textbf{{39.74\%}} & \yg{37.18\%} & 36.54\% & 39.10\% & 39.10\% & 39.10\%($\uparrow 2.56\%$) \\
    & Pass@5 & 44.23\% & \textbf{{51.28\%}} & \yg{47.44\%} & 47.44\% & 47.44\% & 48.08\% & 48.72\% ($\uparrow 4.49\%$)\\
    & Pass@10 & 51.28\% & 51.28\% & \yg{50.64\%} & 50.64\% & \textbf{{53.85\%}} & 50.64\% & 51.28\% ($\uparrow 0.00\%$)\\
    \bottomrule
    \multicolumn{9}{c}{\rule{0pt}{1.5em}\textit{The Results of {\tool} on \syh{CodeLlama 7B}}} \\
    \midrule
    \multirow{3}{*}{ResBench} & Pass@1 & 35.71\% & 35.71\% & \yg{44.64\%} & 44.64\% & 42.86\% & 46.43\% & \textbf{46.43\%($\uparrow 10.72\%$)} \\
    & Pass@5 & 60.71\% & 51.79\% & \yg{60.71\%} & 58.93\% & 58.93\% & 64.29\% & \textbf{{67.86\%}($\uparrow 7.15\%$)} \\
    & Pass@10 & 55.36\% & 62.50\% & \yg{64.29\%} & 62.50\% & 66.07\% & 69.64\% & \textbf{{69.64\%}($\uparrow 14.28\%$)} \\
    \midrule
    \multirow{3}{*}{RTLLM} & Pass@1 & 26.00\% & \textbf{30.00\%} & \yg{24.00\%} & 22.00\% & 20.00\% & 24.00\% & \textbf{{30.00\%}($\uparrow 4.00\%$)} \\
    & Pass@5 & 34.00\% & 36.00\% & \yg{34.00\%} & 33.97\% & 36.00\% & 36.00\% & \textbf{{38.00\%}($\uparrow 4.00\%$)} \\
    & Pass@10 & \textbf{{44.00\%}} & 38.00\% & \yg{40.00\%} & 40.00\% & 40.00\% & 42.00\% & 42.00\% ($\downarrow 2.00\%$)\\
    \midrule
    \multirow{3}{*}{VerilogEval} & Pass@1 & 35.26\% & 32.05\% & \yg{35.26\%} & 33.97\% & 33.97\% & 37.18\% & \textbf{39.74\%($\uparrow 4.48\%$)} \\
    & Pass@5 & 42.95\% & 39.74\% & \yg{42.95\%} & 42.95\% & 43.59\% & 41.67\% & \textbf{{46.15\%}($\uparrow 3.20\%$)} \\
    & Pass@10 & 45.51\% & 42.31\% & \yg{46.15\%} & 44.87\% & \textbf{48.08\%} & 46.15\% & \textbf{{48.08\%}($\uparrow 2.57\%$)} \\
    \bottomrule
    \multicolumn{9}{c}{\rule{0pt}{1.5em}\textit{The Results of {\tool} on SeedCoder 8B}} \\
    \midrule
    \multirow{3}{*}{ResBench} & Pass@1 & 37.50\% & 51.79\% & \yg{53.57\%} & 51.79\% & 55.36\% & 51.79\% & \textbf{{62.50\%}($\uparrow 25.00\%$)} \\
    & Pass@5 & 57.14\% & 71.43\% & \yg{75.00\%} & 75.00\% & 69.64\% & 64.29\% & \textbf{{76.79\%}($\uparrow 19.65\%$)} \\
    & Pass@10 & 66.07\% & 82.14\% & \yg{80.36\%} & 78.57\% & 80.36\% & 73.21\% & \textbf{{85.71\%}($\uparrow 19.64\%$)} \\
    \midrule
    \multirow{3}{*}{RTLLM} & Pass@1 & 36.00\% & \textbf{{42.00\%}} & \yg{38.00\%} & 38.00\% & 38.00\% & 38.00\% & 36.00\% ($\uparrow 0.00\%$)\\
    & Pass@5 & 38.00\% & \textbf{{48.00\%}} & \yg{44.00\%} & 42.00\% & 44.00\% & 44.00\% & 46.00\% ($\uparrow 8.00\%$)\\
    & Pass@10 & 44.00\% & 48.00\% & \yg{50.00\%} & 48.00\% & 50.00\% & 50.00\% & \textbf{{54.00\%}($\uparrow 10.00\%$)} \\
    \midrule
    \multirow{3}{*}{VerilogEval} & Pass@1 & 39.74\% & 46.15\% & \yg{47.44\%} & 46.79\% & \textbf{50.00\%} & 43.29\% & 48.08\% ($\uparrow 8.34\%$)\\
    & Pass@5 & 50.64\% & 51.92\% & \yg{52.56\%} & 51.92\% & 52.56\% & 51.28\% & \textbf{{53.85\%}($\uparrow 3.21\%$)} \\
    & Pass@10 & 53.21\% & 55.13\% & \yg{55.77\%} & 55.77\% & 55.77\% & 56.41\% & \textbf{{56.41\%}($\uparrow 3.20\%$)} \\
\bottomrule
\end{tabular}}
\end{table*}

\subsubsection{Performance Comparison.}
{\tool} outperforms state-of-the-art baselines and full-dataset training in most cases across three mainstream code models (Qwen2.5-Coder 7B, CodeLlama 7B, SeedCoder 8B) and three Verilog benchmarks (VerilogEval-v2, ResBench, RTLLM).

For Qwen2.5-Coder 7B (20\% data), {\tool} achieves Pass@1 scores of 57.14\%, 36.00\%, and 39.10\% on ResBench, RTLLM, and VerilogEval, respectively, yielding an average Pass@1 of 44.08\%. This represents a relative improvement of 18.49\% over full-dataset training (average Pass@1 = 37.20\%) and 1.36\% over the best-performing baseline Cherry (43.49\%).
The most significant gain is observed on ResBench, where Pass@1 improves by 39.13\% relative to full-dataset training.
Notably, {\tool} achieves comparable performance to VeriCoder (which uses 46.55\% of data) with only 20\% data on VerilogEval, demonstrating exceptional data utilization efficiency.

For CodeLlama 7B (25\% data), {\tool} achieves Pass@1 scores of 46.43\%, 30.00\%, and 39.74\% on the three benchmarks, yielding an average Pass@1 of 38.72\%, with a relative improvement of 19.80\% over full-dataset training (32.32\%) and 7.95\% over Cherry (35.87\%).
On ResBench, its Pass@5 (67.86\%) improves by 11.78\% over full-dataset training (60.71\%) and 31.03\% over VeriCoder (51.79\%), demonstrating that {\tool} can significantly enhance the performance of general-purpose code models on Verilog tasks.

For SeedCoder 8B (25\% data), {\tool} exhibits the most significant improvement: Pass@1 scores of 62.50\%, 36.00\%, and 48.08\% yield an average Pass@1 of 48.86\%, improving by 29.43\% over full-dataset training (37.75\%) and 2.24\% over the best baseline EL2N (47.79\%).
On ResBench, its Pass@1 (62.50\%) achieves a 66.67\% relative increase over full-dataset training (37.50\%). On RTLLM, its Pass@10 (54.00\%) improves by 8.00\% over Cherry (50.00\%) and 22.73\% over full-dataset training (44.00\%).

Across all models, {\tool} achieves outstanding Performance Retention Rates (PRR, i.e., the ratio of subset-trained to full-dataset-trained performance): 118.49\% for Qwen2.5-Coder 7B (20\% data), 119.80\% for CodeLlama 7B (25\% data), and 129.43\% for SeedCoder 8B (25\% data) in average Pass@1. 

\yg{To further isolate the contribution of {\tool}'s IFD scoring, structural features, and diversity-aware clustering beyond testbench filtering, we introduce VeriCoder$_R$ as a controlled baseline that randomly subsamples from VeriCoder's quality-filtered pool to match the target DSR. Comparing VeriCoder$_R$ against {\tool} under the same data budget reveals a consistent performance gap: on Qwen2.5-Coder 7B, VeriCoder$_R$ achieves average Pass@1 of 41.58\% (53.57\%, 34.00\%, 37.18\%) compared to {\tool}'s 44.08\%, a 2.50 percentage point improvement attributable entirely to the subsequent pipeline stages. This gap persists across CodeLlama 7B (VeriCoder$_R$: 34.63\% vs.\ Ours: 38.72\%) and SeedCoder 8B (VeriCoder$_R$: 46.34\% vs.\ Ours: 48.86\%), confirming that quality filtering alone is insufficient and that the diversity-aware selection mechanism provides additional gains beyond what testbench verification can offer.}

\subsubsection{Distribution Characteristics}

\begin{figure}[htb]
	\centering
    \vspace{-1mm}
    \includegraphics[width=1.0\linewidth]{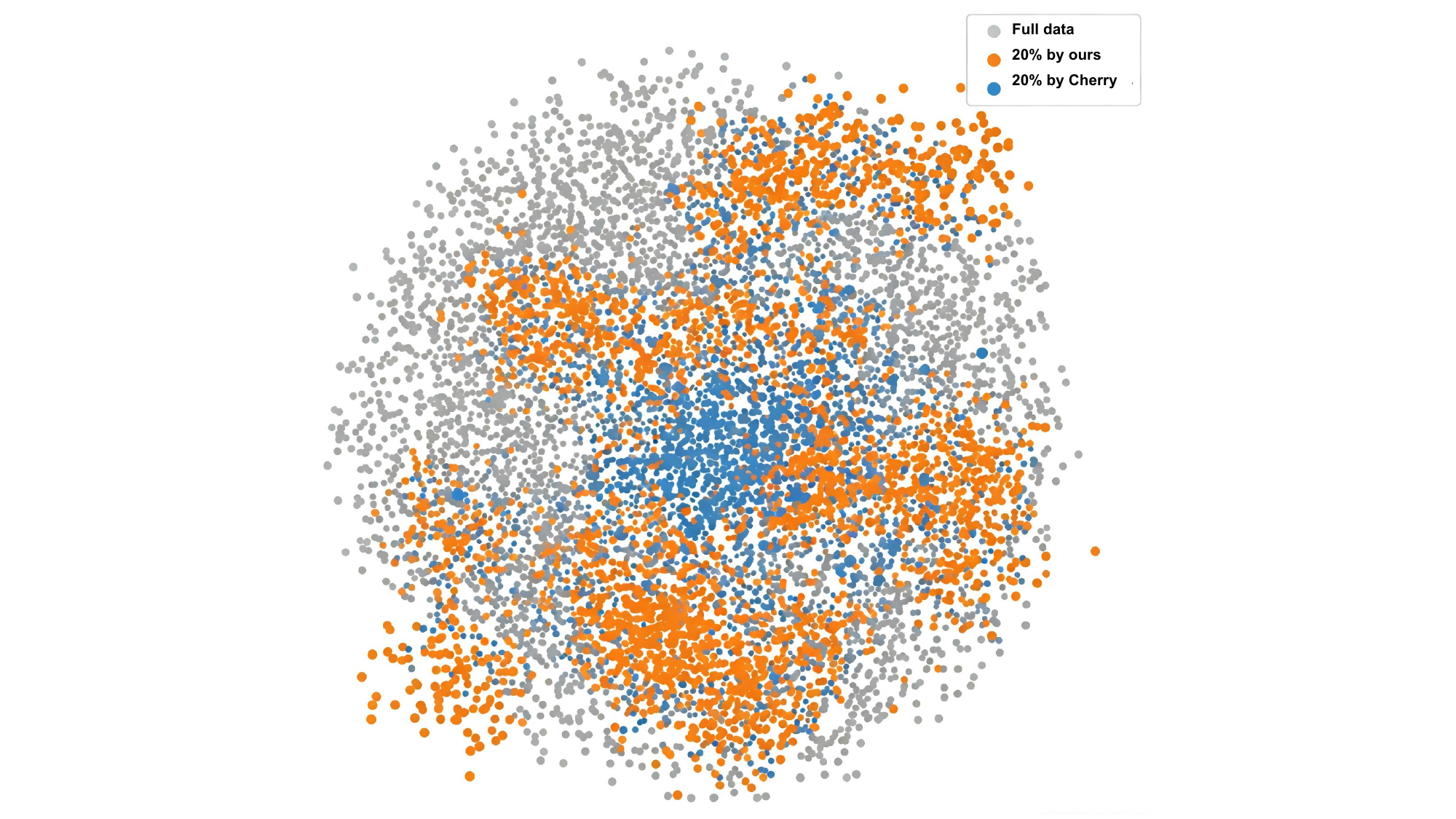}
	\caption{Visualization using t-SNE on instruction-solution embeddings from RTLCoder dataset.}
    \vspace{-1mm}
	\label{fig:tsne}
\end{figure} 

To evaluate the distribution characteristics of the selected 20\% subset within the full dataset, we compute the semantic embeddings of instructions and solutions, and project them into a 2D space via t-SNE for intuitive comparison. 
Figure~\ref{fig:tsne} presents the results of data visualization.
As shown in the figure, the orange points denote the data sampled by {\tool}, while the blue points represent the data selected by the state-of-the-art model Cherry. In contrast to Cherry, the samples selected by {\tool} are not confined to a single region but are uniformly distributed across all semantic clusters discernible in the visualization.

These results empirically confirm that {\tool} effectively preserves the global diversity of the original dataset.
This observation is supported by Proposition~\ref{prop:dist}: with $K{=}100$ clusters and $N_{\text{target}}{\approx}5{,}500$, the $L_1$ distributional deviation between the selected subset and the quality-filtered pool is bounded by $2K/N_{\text{target}} \approx 0.036$. In other words, proportional sampling guarantees less than 4\% distributional shift. Cherry lacks such a proportional allocation mechanism and thus over-samples high-density semantic regions, as reflected by its concentrated distribution in Figure~\ref{fig:tsne}.

\yg{\subsubsection{Quantitative Diversity: Cluster Coverage}}

\yg{While t-SNE provides intuitive evidence, we further quantify diversity using \emph{cluster coverage}, defined as the fraction of the $K{=}100$ clusters that contain at least one selected sample. A higher coverage indicates that the selected subset spans a broader range of structural design patterns.}

\yg{\begin{table}[htbp]
\centering
\caption{Cluster coverage (\%) of different methods at DSR = 20\%.}
\label{tab:cluster_coverage}
\begin{tabular}{lc}
\toprule
\textbf{Method} & \textbf{Coverage (\%)} \\
\midrule
EL2N    & 72.00 \\
Cherry  & 78.00 \\
VeriCoder$_R$ & 88.00 \\
Ours    & \textbf{100.00} \\
\bottomrule
\end{tabular}
\end{table}}

\yg{As shown in Table~\ref{tab:cluster_coverage}, {\tool} achieves 100\% cluster coverage by design: the proportional sampling mechanism first reserves one sample per cluster before distributing the remainder, guaranteeing that every cluster is represented in the selected subset. VeriCoder$_R$ achieves 88\% coverage: although its quality-filtered pool removes erroneous samples, some small clusters lose all members during filtering, leaving 12 clusters uncovered. Cherry and EL2N exhibit substantially lower coverage (78\% and 72\%, respectively) because their quality-oriented scoring concentrates on high-scoring samples that tend to reside in common design pattern clusters, systematically neglecting rare but structurally distinct clusters. The 28\% coverage gap between EL2N and {\tool} directly translates to the performance advantages observed in Table~\ref{tab:RQ1}: missing clusters correspond to underrepresented design patterns that the model fails to learn.}

\subsubsection{Significance Testing}
To rigorously verify that the performance improvement of {\tool} is not attributable to random error, we perform statistical analysis on model-generated results using McNemar's test~\cite{mcnemar1947note}. 
This test is applicable to paired binary outcomes and aims to compare sample differences where one method succeeds while the other fails.

The test statistic $\chi^2$ is calculated as follows:
\begin{equation}
\nonumber
\chi^2 = \frac{(|n_{01} - n_{10}| - 1)^2}{n_{01} + n_{10}},
\end{equation}
where $n_{01}$ denotes the number of samples for which Cherry predicts correctly but {\tool} predicts incorrectly, and $n_{10}$ denotes the number of samples for the opposite case.

\begin{table}[t]
\centering
\caption{McNemar's test comparing {\tool} against Cherry on Pass@1 (SeedCoder 8B).}
\label{tab:mcnemar}
\begin{tabular}{cccccc}
\toprule
\textbf{Benchmark} & $\boldsymbol{n_{00}}$ & $\boldsymbol{n_{01}}$ & $\boldsymbol{n_{10}}$ & $\boldsymbol{n_{11}}$ & \textbf{$p$-value} \\
\midrule
ResBench    & 29 & 3 & 16 & 8 & 0.0059 \\
RTLLM       & 22 & 3 & 17 & 8 & 0.0036\\
VerilogEval & 82 & 9 & 31 & 34 & 0.0009 \\
\bottomrule
\end{tabular}
\end{table}

Compared with the best-performing general-purpose baseline Cherry, {\tool} achieves statistically significant performance improvements across all three benchmarks.
As shown in Table~\ref{tab:mcnemar}, consistent trends of $n_{10}>n_{01}$ are observed on ResBench, RTLLM, and VerilogEval, with the corresponding $p$-values all below 0.01.
These results demonstrate that the advantages of {\tool} over Cherry are not attributable to random variation but are statistically significant.

\begin{tcolorbox}[width=1.0\linewidth, title={Answer to RQ1}]
\yg{{\tool} consistently outperforms state-of-the-art baselines on most evaluation scenarios, even under unified data budgets. Quantitative diversity analysis confirms that {\tool} achieves superior cluster coverage, and statistical significance testing validates the reliability of the improvements.}
\end{tcolorbox}

\begin{table*}[htbp]
\centering

\caption{The Results of {\tool} on Different Models Under Different Data Selection Rates.}
\label{tab:RQ2}
\resizebox{0.8\textwidth}{!}{
\begin{tabular}{ccccccccc}
\toprule
\textbf{Benchmark} & \textbf{Pass@k} & \textbf{100.00\%} & \textbf{46.55\%} & \textbf{25.00\%} & \textbf{20.00\%} & \textbf{15.00\%} & \textbf{10.00\%} & \textbf{5.00\%}  \\
    \bottomrule
    \multicolumn{9}{c}{\rule{0pt}{1.5em}\textit{The Results of {\tool} on \syh{Qwen2.5-Coder 7B} Under Different Data Selection Rates}} \\
    \midrule
    \multirow{3}{*}{ResBench} & Pass@1 & 41.07\% & 55.36\% & 55.36\% & 57.14\% & \textbf{{58.93\%}} & 57.14\% & 55.36\% \\
    & Pass@5 & 62.50\% & 69.64\% & \textbf{{71.43\%}} & 69.64\% & 67.86\% & 69.64\% & 62.50\% \\
    & Pass@10 & 69.64\% & 69.64\% & 69.64\% & 71.43\% & 73.21\% & \textbf{{75.00\%}} & 71.43\% \\
    \midrule
    \multirow{3}{*}{RTLLM} & Pass@1 & 34.00\% & \textbf{{36.00\%}} & \textbf{{36.00\%}} & \textbf{{36.00\%}} & 34.00\% & \textbf{{36.00\%}} & 30.00\% \\
    & Pass@5 & 40.00\% & 44.00\% & 40.00\% & \textbf{{46.00\%}} & 44.00\% & 40.00\% & 40.00\% \\
    & Pass@10 & 40.00\% & 46.00\% & 48.00\% & \textbf{{50.00\%}} & 42.00\% & 42.00\% & \textbf{50.00\%} \\
    \midrule
    \multirow{3}{*}{VerilogEval} & Pass@1 & 36.54\% & \textbf{39.74\%} & \textbf{{39.74\%}} & 39.10\% & 37.82\% & 37.18\% & 36.54\% \\
    & Pass@5 & 44.23\% & \textbf{{51.28\%}} & 48.08\% & 48.72\% & 46.15\% & 44.87\% & 44.87\% \\
    & Pass@10 & \textbf{51.28\%} & \textbf{51.28\%} & 49.36\% & \textbf{{51.28\%}} & 46.79\% & 49.36\% & 48.72\% \\
    \bottomrule
    \multicolumn{9}{c}{\rule{0pt}{1.5em}\textit{The Results of {\tool} on \syh{CodeLlama 7B} Under Different Data Selection Rates}} \\
    \midrule
    \multirow{3}{*}{ResBench} & Pass@1 & 35.71\% & 35.71\% & 46.43\% & \textbf{50.00\%} & 48.21\% & 46.43\% & 46.43\% \\
    & Pass@5 & 60.71\% & 51.79\% & \textbf{{67.86\%}} & 57.14\% & 58.93\% & 55.36\% & \textbf{67.86\%} \\
    & Pass@10 & 55.36\% & 62.50\% & \textbf{{69.64\%}} & \textbf{69.64\%} & 62.50\% & 64.29\% & 64.29\% \\
    \midrule
    \multirow{3}{*}{RTLLM} & Pass@1 & 26.00\% & \textbf{30.00\%} & \textbf{{30.00\%}} & 26.00\% & \textbf{30.00\%} & 22.00\% & 22.00\% \\
    & Pass@5 & 34.00\% & 36.00\% & \textbf{{38.00\%}} & 34.00\% & 36.00\% & 36.00\% & 32.00\% \\
    & Pass@10 & \textbf{44.00\%} & 38.00\% & 42.00\% & 38.00\% & 38.00\% & 42.00\% & 34.00\% \\
    \midrule
    \multirow{3}{*}{VerilogEval} & Pass@1 & 35.26\% & 32.05\% & \textbf{{39.74\%}} & 38.46\% & 33.97\% & 33.33\% & 34.62\% \\
    & Pass@5 & 42.95\% & 39.74\% & \textbf{{46.15\%}} & 44.87\% & 42.95\% & 44.23\% & 39.10\% \\
    & Pass@10 & 45.51\% & 42.31\% & \textbf{{48.08\%}} & \textbf{48.08\%} & 46.79\% & 44.87\% & 44.87\% \\
    \bottomrule
    \multicolumn{9}{c}{\rule{0pt}{1.5em}\textit{The Results of {\tool} on \syh{SeedCoder 8B} Under Different Data Selection Rates}} \\
    \midrule
    \multirow{3}{*}{ResBench} & Pass@1 & 37.50\% & 51.79\% & \textbf{{62.50\%}} & \textbf{62.50\%} & 55.36\% & 53.57\% & 50.00\% \\
    & Pass@5 & 57.14\% & 71.43\% & 76.79\% & 76.79\% & \textbf{80.36\%} & 75.00\% & 75.00\% \\
    & Pass@10 & 66.07\% & 82.14\% & \textbf{{85.71\%}} & 78.57\% & 80.36\% & 80.36\% & 78.57\% \\
    \midrule
    \multirow{3}{*}{RTLLM} & Pass@1 & 36.00\% & \textbf{42.00\%} & 36.00\% & 36.00\% & 38.00\% & 36.00\% & 32.00\% \\
    & Pass@5 & 38.00\% & \textbf{48.00\%} & 46.00\% & 42.00\% & \textbf{48.00\%} & \textbf{48.00\%} & 46.00\% \\
    & Pass@10 & 44.00\% & 48.00\% & \textbf{{54.00\%}} & \textbf{54.00\%} & \textbf{54.00\%} & 52.00\% & 48.00\% \\
    \midrule
    \multirow{3}{*}{VerilogEval} & Pass@1 & 39.74\% & 46.15\% & \textbf{{48.08\%}} & 44.87\% & 43.59\% & 46.15\% & \textbf{48.08\%} \\
    & Pass@5 & 50.64\% & 51.92\% & 53.85\% & \textbf{55.13\%} & 51.92\% & 52.56\% & 53.85\% \\
    & Pass@10 & 53.21\% & 55.13\% & \textbf{{56.41\%}} & 53.21\% & 55.13\% & 55.77\% & 55.77\% \\
    \bottomrule
\end{tabular}}
\end{table*}

\subsection{RQ2: \syh{Sensitivity Analysis}}

To investigate the impact of DSR on the performance of {\tool}, we tested DSR values ranging from 5\% to 100\% on the three models, with results presented in Table~\ref{tab:RQ2}.

\textbf{Full dataset is not optimal.} Training on all data (DSR=100\%) is not the optimal choice, as all three models perform significantly worse on full data than on the selected subsets. 
For Qwen2.5-Coder 7B, the average Pass@1 on full data is 37.20\%, which is 6.88 percentage points lower than the optimal value of 44.08\% at 20\% DSR (a relative improvement of 18.49\%).
For CodeLlama 7B, the average Pass@1 on full data is 32.32\%, 6.40 percentage points lower than the optimal value of 38.72\% at 25\% DSR (a relative improvement of 19.80\%).
For SeedCoder 8B, the average Pass@1 on full data is 37.75\%, 11.11 percentage points lower than the optimal value of 48.86\% at 25\% DSR (a relative improvement of 29.43\%). 
This confirms that redundant information in the original dataset interferes with the model's learning of core Verilog design knowledge.

\textbf{Inverted U-shaped relationship.} DSR exhibits an inverted U-shaped relationship with model performance, with performance peaks concentrated in the 15\%--25\% range. 
Figure~\ref{fig:dsr} presents the trend.
Qwen2.5-Coder 7B achieves optimal performance at 20\% DSR, with an average Pass@1 of 44.08\%.
Both CodeLlama 7B and SeedCoder 8B perform best at 25\% DSR, with average Pass@1 values of 38.72\% and 48.86\%, respectively. 
Notably, SeedCoder 8B attains a Pass@1 of 62.50\% on the ResBench benchmark at 25\% DSR, representing a 66.67\% relative increase over full-dataset training. 
Even at an extremely low DSR of 5\%, all three models still outperform full-dataset training (by 9.22\%, 6.28\%, and 4.86\% in relative terms for Qwen2.5-Coder 7B, CodeLlama 7B, and SeedCoder 8B, respectively), validating {\tool}'s ability to select high-value samples.
However, DSR below 10\% leads to a slight performance decline due to insufficient data coverage.

Considering both performance gains and data efficiency, the optimal configurations are: 20\% DSR for Qwen2.5-Coder 7B (PRR=118.49\%), and 25\% DSR for CodeLlama 7B and SeedCoder 8B (PRR of 119.80\% and 129.43\%, respectively). 
This configuration reduces data requirements by 75\%--80\% while achieving relative performance improvements of 18.49\%--29.43\%.

\textbf{Theoretical interpretation.}
The inverted U-shaped trend is consistent with the classical bias-noise trade-off in supervised learning~\cite{domingos2000unified}. At low DSR, the selected subset has insufficient coverage of hardware design patterns, leading to underfitting. At high DSR, redundant and misaligned samples accumulate, degrading the signal-to-noise ratio of the training set. The optimal DSR represents the point where the marginal coverage gain is offset by the marginal noise intake. This also explains why full-dataset training (DSR=100\%) consistently underperforms: the RTLCoder-27k dataset, constructed via automated keyword expansion, contains non-trivial noise that penalizes models when included without filtering.

\begin{figure*}[t]
  \centering
  \begin{subfigure}[t]{0.32\textwidth}
    \centering
    \includegraphics[width=\linewidth]{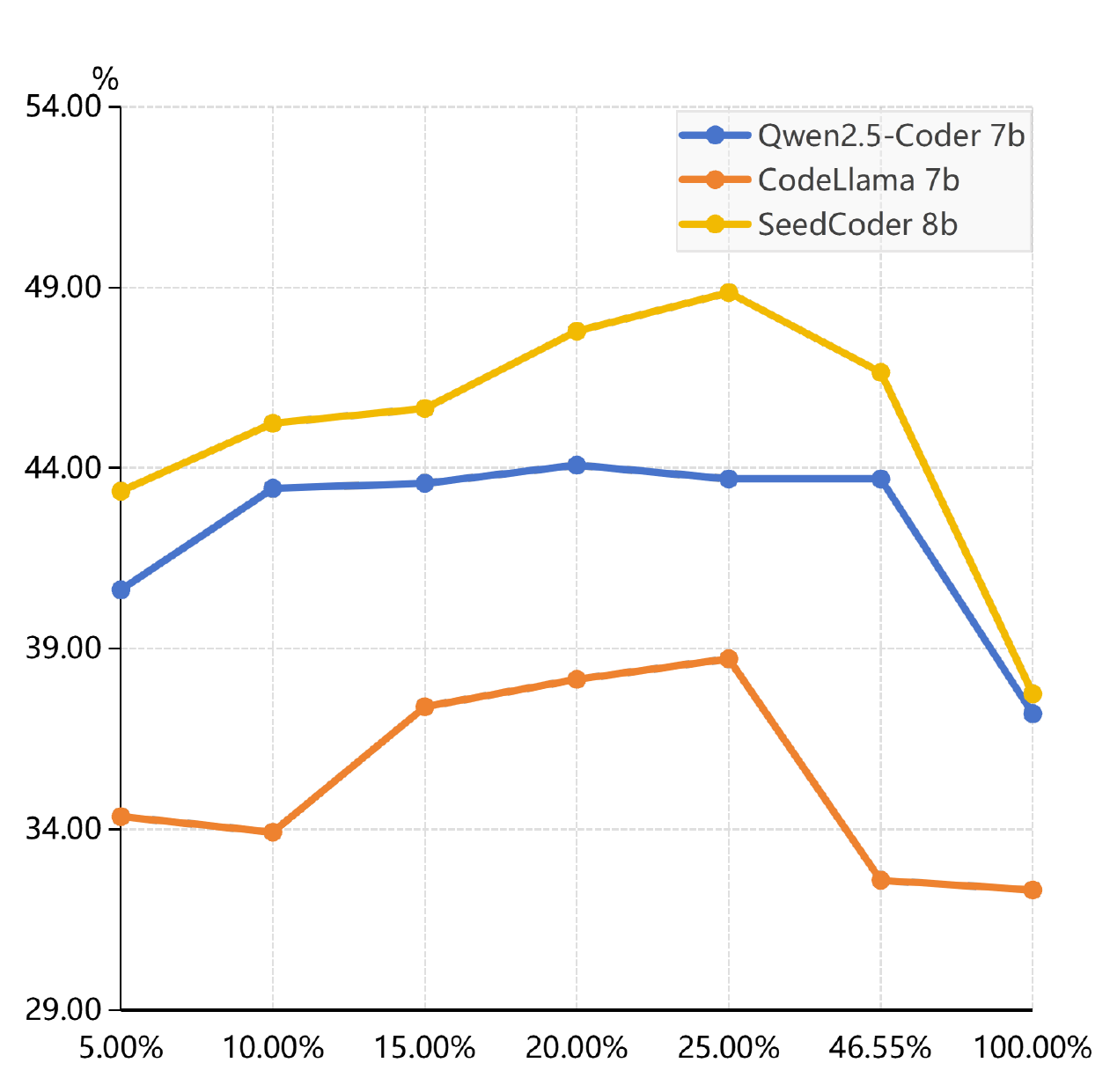}
    \caption{DSR on Avg Pass@1}
    \label{fig:dsr-pass1}
  \end{subfigure}\hfill
  \begin{subfigure}[t]{0.32\textwidth}
    \centering
    \includegraphics[width=\linewidth]{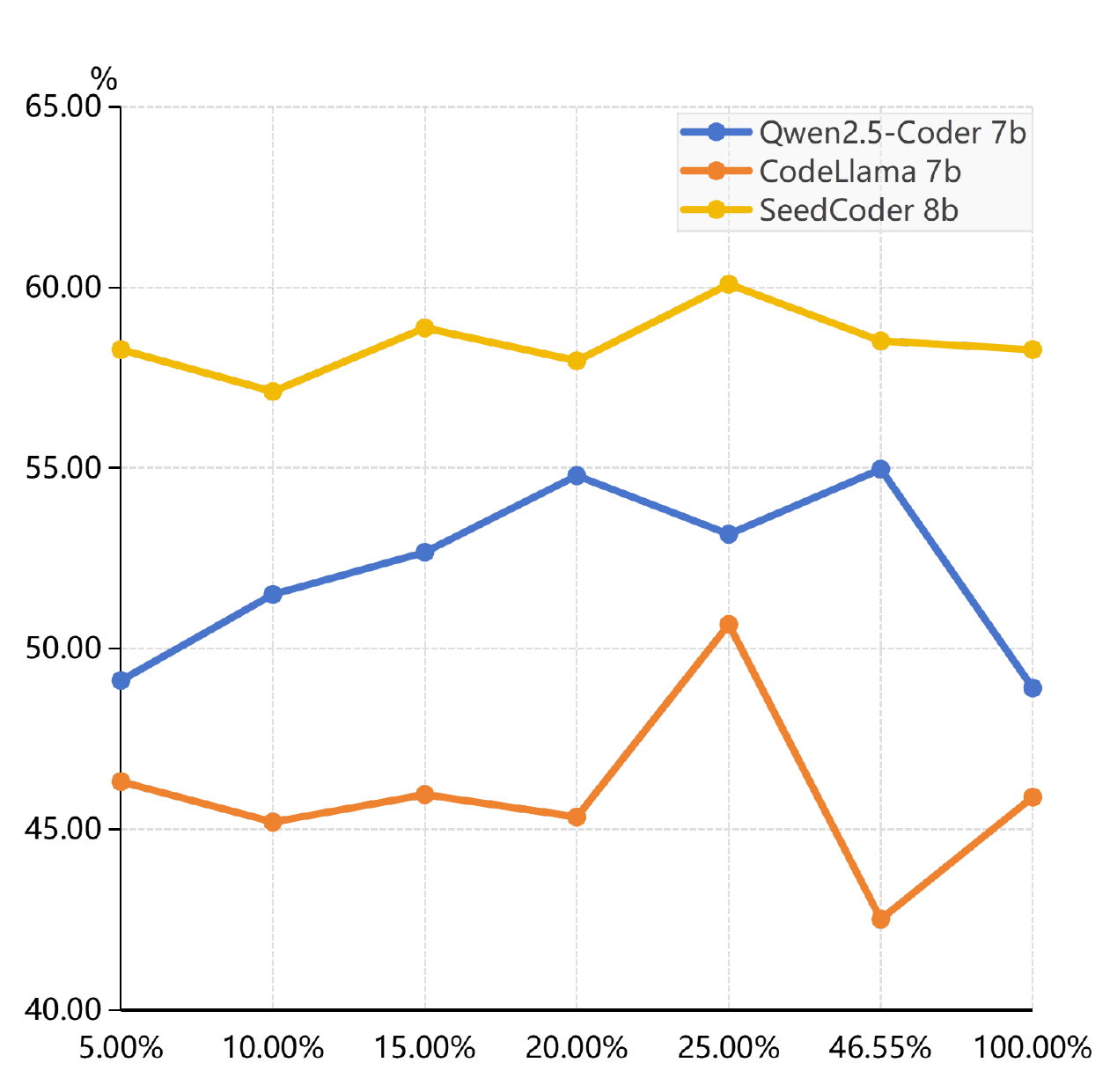}
    \caption{DSR on Avg Pass@5}
    \label{fig:dsr-pass5}
  \end{subfigure}\hfill
  \begin{subfigure}[t]{0.32\textwidth}
    \centering
    \includegraphics[width=\linewidth]{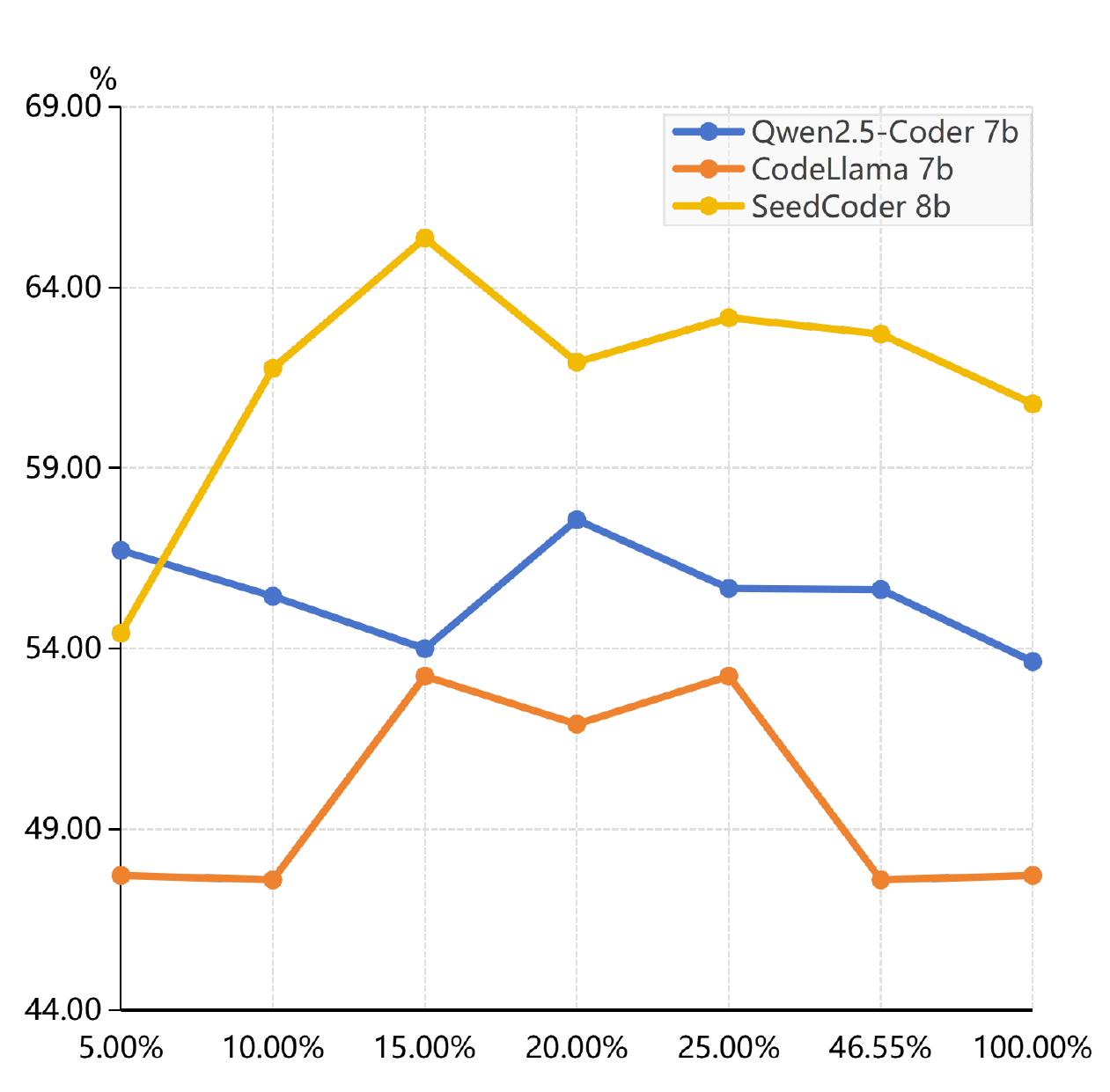}
    \caption{DSR on Avg Pass@10}
    \label{fig:dsr-pass10}
  \end{subfigure}
  \caption{Data selection rate analysis under different evaluation metrics.}
  \label{fig:dsr}
\end{figure*}

\begin{tcolorbox}[width=1.0\linewidth, title={Answer to RQ2}]
The performance of {\tool} exhibits an inverted U-shaped trend with increasing DSR, with the optimal range at 20\%--25\%.
This configuration maximizes performance gains while reducing training data volume, validating the effectiveness of the selection method.
\end{tcolorbox}

\subsection{RQ3: Component Analysis}

\begin{figure*}[htb]
	\centering
    \includegraphics[width=\textwidth]{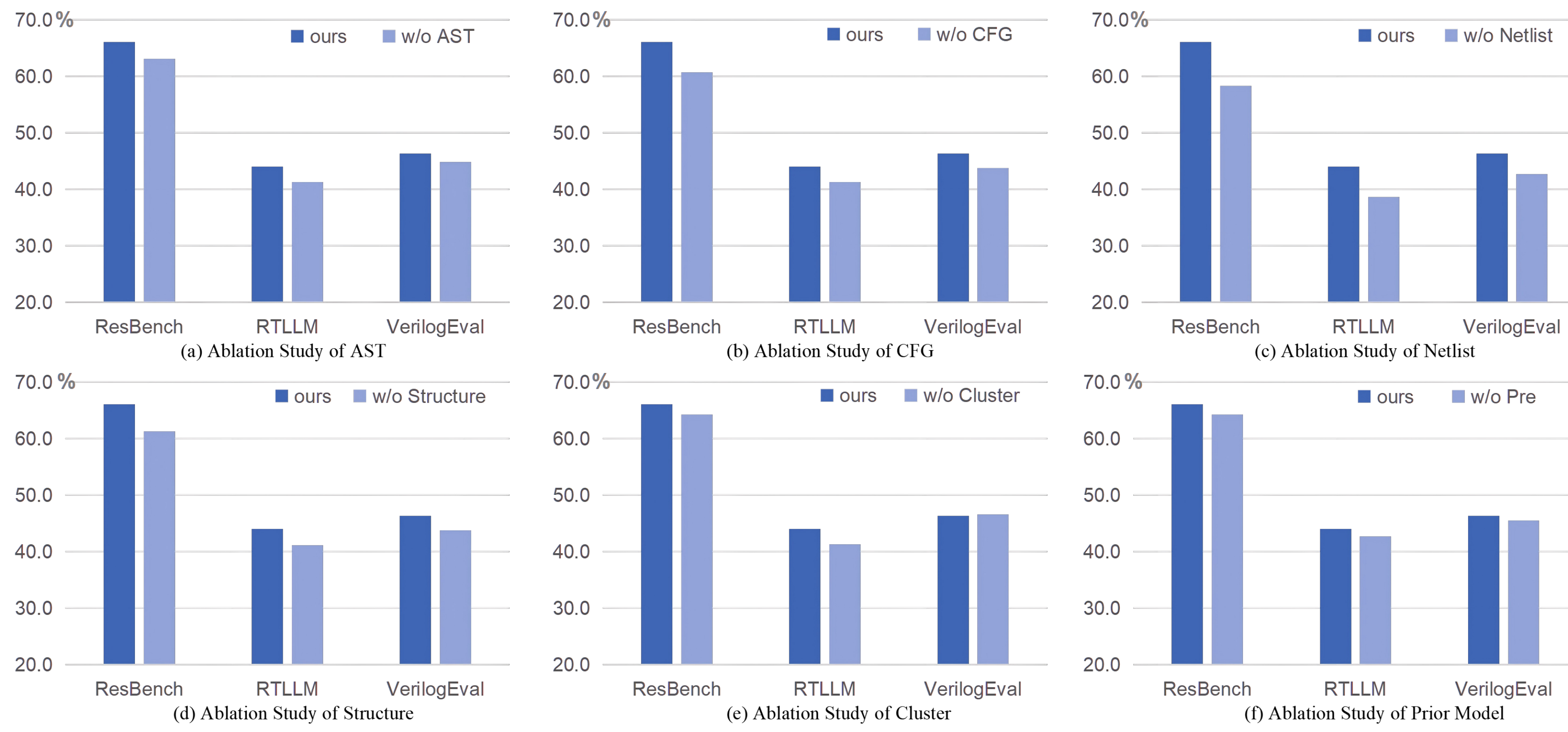}
	\caption{Ablation studies of different components.}
	\label{fig:ablation}
\end{figure*} 

To systematically validate the contributions of each component in {\tool}, we conduct ablation studies on three benchmarks (ResBench, RTLLM, and VerilogEval) with average Pass@k, as illustrated in Figure~\ref{fig:ablation}. 
Note that the VeriCoder baseline in RQ1 already serves as a natural ablation of the testbench verification component, as it retains only functional correctness filtering without IFD scoring, structural features, or diversity-aware clustering. The results in Table~\ref{tab:RQ1} show that {\tool} consistently outperforms VeriCoder in average performance while using substantially less data (20\%--25\% vs.\ 46.55\%), confirming the added value of the subsequent pipeline stages. Below, we focus on finer-grained ablations of the remaining components.

\textbf{(a) AST.}
Removing the AST leads to drops in average Pass@k of 4.50\%, 6.07\%, and 3.23\% on ResBench, RTLLM, and VerilogEval, respectively. 
This demonstrates that the hierarchical syntactic representation provided by the AST serves as the foundation for the model to master Verilog syntax rules and generate structurally well-formed code.

\textbf{(b) CFG.}
Ablating the CFG results in average Pass@k decreases of 8.10\%, 6.07\%, and 5.54\% across the three benchmarks, highlighting the critical role of CFG in modeling execution-flow logic, which is especially important for control-flow-intensive tasks such as ResBench.

\textbf{(c) Netlist.}
Notably, removing the Netlist causes the most severe single-component degradation: average Pass@k drops by 11.72\%, 12.11\%, and 7.85\% on ResBench, RTLLM, and VerilogEval.
This directly validates that Netlist, as a fundamental hardware structural representation, plays a core role in capturing Verilog semantics and ensuring generated code complies with physical implementation constraints.

\textbf{(d) Structure.}
When the entire structural module (AST, CFG, and Netlist combined) is removed, average Pass@k reductions are 7.20\%, 6.57\%, and 5.54\% on the three benchmarks. Notably, these drops are \emph{smaller} than those caused by removing Netlist alone (11.72\%, 12.11\%, and 7.85\%). We attribute this to a \textbf{feature space consistency} effect. Removing only Netlist leaves AST and CFG in the fusion, but the resulting feature space lacks the hardware-implementation semantics that Netlist encodes. K-Means then clusters in this incomplete space, producing distorted partitions that misgroup structurally distinct samples. The corrupted cluster structure propagates through proportional sampling, ultimately degrading training quality. In contrast, removing all structural features causes the pipeline to fall back to a pure textual embedding space. Although this space loses all structural information, it is internally consistent: clustering operates on a coherent representation, and proportional sampling remains meaningful. This consistency partially compensates for the information loss, resulting in a smaller overall degradation. The finding highlights that \emph{feature space coherence matters more than feature volume} in fusion-based selection pipelines.

\textbf{(e) Cluster.}
To investigate the impact of the diversity clustering strategy on model capability, we ablate the Cluster module. 
It results in Pass@k drops of 2.69\% and 6.07\% on ResBench and RTLLM, respectively.
This indicates that the diversity strategy groups structurally similar code samples into coherent clusters, helping the model capture Verilog-specific patterns and improve adaptability to diverse scenarios. 
However, on VerilogEval, average Pass@k increases slightly by 0.45\%, a negligible change. 
A possible explanation is that VerilogEval's problem distribution is relatively concentrated, and the clustering strategy provides limited additional diversity benefit in this scenario. 
Such dataset-dependent effects reveal that diversity clustering is a context-aware strategy, whose effectiveness depends on the complexity and diversity characteristics of the target benchmark.

\textbf{(f) Prior Model.}
To evaluate the contribution of the prior model, we ablate the Prior Model module. After removal, average Pass@k decreases by 2.69\%, 3.02\%, and 1.85\% on ResBench, RTLLM, and VerilogEval, respectively.
Although the performance gain is smaller than that of the structural information module, the consistent performance drop across all benchmarks confirms that the prior model provides more domain-aware embeddings for both IFD scoring and feature extraction, contributing to more effective quality assessment and diversity-aware selection.


\yg{\subsubsection{Feature Importance Analysis}}

\yg{To further validate that the 109-dimensional structural features encode meaningful and non-redundant information, we conduct a feature importance analysis using Random Forest~\cite{breiman2001random}. We train a Random Forest classifier to predict which cluster each sample belongs to, using the structural features as input, and measure the importance of each feature dimension via mean decrease in impurity.}

\yg{Figure~\ref{fig:feature_importance} presents the feature importance scores grouped by feature category. The results reveal three key observations. First, high-importance features are distributed across all three groups (AST, CFG, and Netlist), confirming that each group captures distinct and complementary structural information rather than redundant signals. Second, Netlist features exhibit the highest average importance, consistent with the ablation finding that removing Netlist causes the most severe performance degradation. Third, the importance distribution is non-uniform within each group, indicating that certain structural dimensions (e.g., gate type distribution in Netlist, branching density in CFG) are particularly informative for distinguishing hardware design patterns.}

\yg{\begin{figure*}[htbp]
\centering
\includegraphics[width=\linewidth]{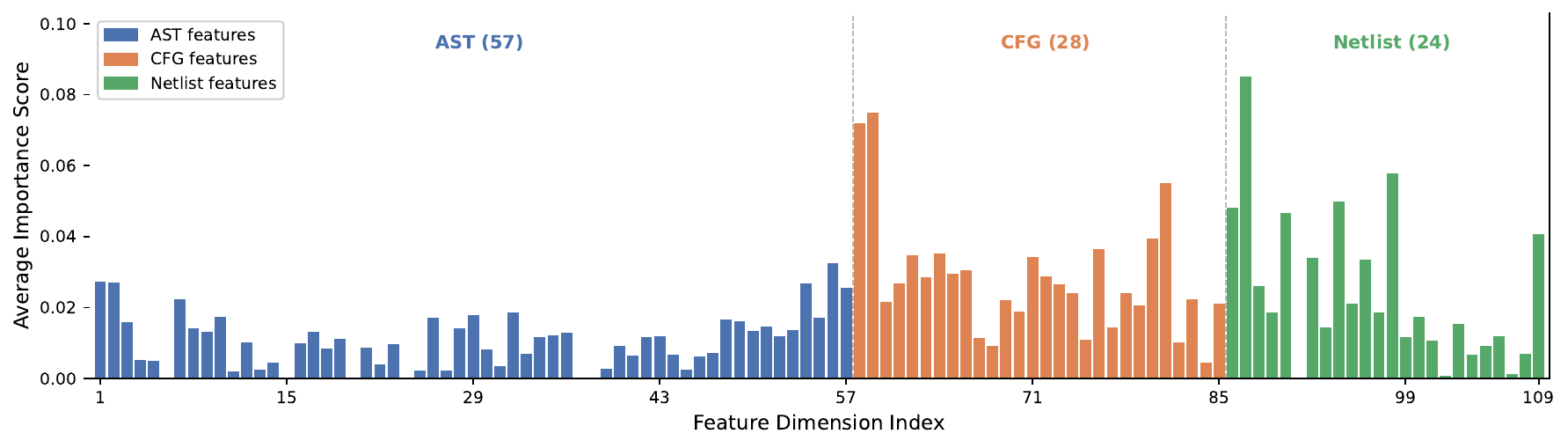}
\caption{Feature importance scores of the 109 structural dimensions, grouped by AST (blue), CFG (orange), and Netlist (green). Higher values indicate greater contribution to cluster discrimination.}
\label{fig:feature_importance}
\end{figure*}}


\yg{\subsubsection{Comparison with Alternative Code Representations}}

\yg{A natural question is whether the proposed 109-dimensional handcrafted structural features can be replaced by learned code representations. To investigate this, we consider two categories of alternatives: pretrained code embeddings and graph neural networks (GNNs).}

\yg{\textbf{Pretrained code embeddings.} We compare against CodeBERT~\cite{feng2020codebert} and UniXcoder~\cite{guo2022unixcoder}, two widely adopted pretrained code models. For each Verilog sample, we feed the solution code into the pretrained encoder and extract the \texttt{[CLS]} token representation as the code embedding vector ($d{=}768$ for both models). This embedding replaces the 109-dimensional structural feature vector $\hat{\boldsymbol{f}}_i^{\text{struct}}$ in the fusion step (Eq.~5): after Z-score normalization, it is concatenated with the textual embedding $\hat{\boldsymbol{h}}_i^Q$ to form the fused representation for downstream clustering and adaptive sampling. All other pipeline stages remain identical.}

\yg{\begin{table}[htbp]
\centering
\caption{Comparison of structural representations on Qwen2.5-Coder 7B (DSR = 20\%). ``Ours'' denotes the 109-dimensional handcrafted features. Average Pass@1 (\%) is reported.}
\label{tab:repr_comparison}
\begin{tabular}{lcccc}
\toprule
\textbf{Method} & \textbf{Dim.} & \textbf{ResBench} & \textbf{RTLLM} & \textbf{VerilogEval} \\
\midrule
CodeBERT    & 768 & 51.79\% & 32.00\% & 35.90\% \\
UniXcoder   & 768 & 53.57\% & 34.00\% & 37.18\% \\
Ours & 109 & \textbf{57.14\%} & \textbf{36.00\%} & \textbf{39.10\%} \\
\bottomrule
\end{tabular}
\end{table}}

\yg{As shown in Table~\ref{tab:repr_comparison}, our handcrafted structural features consistently outperform both pretrained code embeddings despite having substantially lower dimensionality (109 vs.\ 768). Specifically, our features improve average Pass@1 over CodeBERT by 5.35, 4.00, and 3.20 percentage points on ResBench, RTLLM, and VerilogEval, respectively. UniXcoder narrows the gap slightly owing to its cross-modal pretraining objective, yet still falls short by 3.57, 2.00, and 1.92 percentage points. These results indicate that the advantage of our approach stems not from increased feature dimensionality but from the hardware-grounded semantics encoded by AST, CFG, and Netlist features. Pretrained code embeddings are trained on general-purpose programming languages and primarily capture token-level and syntactic patterns; they lack explicit representations of gate-level topology, signal connectivity, and module hierarchy that are essential for distinguishing Verilog design patterns.}

\yg{\textbf{Why not GNN-based representations?} An alternative approach is to apply graph neural networks directly on AST, CFG, or Netlist graphs to learn node-level or graph-level embeddings. We deliberately chose handcrafted features over GNN embeddings for three reasons. First, GNN training requires a task-specific supervised objective (e.g., graph classification labels), which is unavailable in our data selection setting where no ground-truth quality or diversity labels exist. Second, Verilog programs exhibit extreme variability in graph size: the AST of a simple combinational module may contain tens of nodes, while a complex sequential controller can produce thousands. This variability causes severe mini-batch imbalance and padding overhead for GNN training, and the resulting embeddings are sensitive to graph size rather than structural semantics. Third, our handcrafted features are interpretable and hardware-grounded: each dimension corresponds to a well-defined structural property (e.g., gate type distribution, branching density, module hierarchy depth), enabling the feature importance analysis presented above. GNN embeddings, by contrast, lack such interpretability. The empirical results in Table~\ref{tab:repr_comparison} further validate that domain-specific handcrafted features, despite their lower dimensionality, outperform general-purpose learned representations for Verilog data selection.}

\yg{\subsubsection{Analysis of Contributing Factors}}
\label{sec:why}

\yg{The consistent superiority of {\tool} across models and benchmarks stems from three contributing factors that operate synergistically.
First, quality filtering removes noise at complementary granularities: testbench verification eliminates functionally incorrect samples, while IFD scoring discards semantically misaligned instruction-solution pairs. Their sequential application ensures that the surviving samples are both behaviorally correct and pedagogically effective.
Second, structural features capture hardware semantics invisible to textual embeddings. The pronounced contribution of Netlist features in the ablation study and feature importance analysis above reflects the fundamental property that Verilog programs are ultimately synthesized into physical circuits. Netlist features encode gate-level topology and signal connectivity that textual embeddings cannot express, which also explains why {\tool} achieves the largest gains on SeedCoder 8B.
Third, proportional sampling prevents distributional collapse. Without diversity-aware allocation, quality-only selection tends to over-represent common design patterns. The proportional mechanism, backed by Proposition~\ref{prop:dist}, bounds the distributional deviation to less than 4\%, ensuring that rare but important design patterns are retained.}

\begin{tcolorbox}[width=1.0\linewidth, title={Answer to RQ3}]
\yg{All components contribute positively to the overall performance. The VeriCoder baseline validates the testbench verification stage, while finer-grained ablations confirm the necessity of structural features, diversity clustering, and the prior model. Feature importance analysis demonstrates that AST, CFG, and Netlist features capture complementary structural information, and comparison with pretrained code embeddings validates the advantage of domain-specific handcrafted features.} The three contributing factors (quality filtering, structural semantics, proportional sampling) operate synergistically to maximize the signal-to-noise ratio and distributional coverage of the training subset.
\end{tcolorbox}

%% file: sections/6.discussion.tex
\section{Discussion}
\label{sec:discussion}

\subsection{Relationship Analysis}

To explore the relationship between model parameter size and the optimal data selection rate (DSR), we extend our sensitivity analysis to a smaller-scale model: Qwen2.5-Coder 1.5B.
Combined with the results of the three larger models reported in Table~\ref{tab:RQ2}, we compare the optimal DSR across four models spanning 1.5B to 8B parameters.

As shown in Table~\ref{tab:discussion1}, Qwen2.5-Coder 1.5B achieves its best performance at a DSR of 15\%, with average Pass@1/5/10 scores of 41.66\%, 52.60\%, and 57.88\%, respectively.
This represents a 51.03\% relative improvement in average Pass@1 over full-dataset training (27.59\%).
In contrast, the three larger models in Table~\ref{tab:RQ2} require higher DSRs to reach their performance peaks: Qwen2.5-Coder 7B peaks at 20\% DSR (average Pass@1 = 44.08\%), while both CodeLlama 7B and SeedCoder 8B peak at 25\% DSR (average Pass@1 = 38.72\% and 48.86\%, respectively).

\begin{table}[htbp]
\centering
\caption{Summary of optimal DSR across models of different sizes.}
\label{tab:dsr-summary}
\begin{tabular}{ccc}
\toprule
\textbf{Model} & \textbf{Parameters} & \textbf{Optimal DSR} \\
\midrule
Qwen2.5-Coder 1.5B & 1.5B & 15\% \\
Qwen2.5-Coder 7B   & 7B   & 20\% \\
CodeLlama 7B        & 7B   & 25\% \\
SeedCoder 8B        & 8B   & 25\% \\
\bottomrule
\end{tabular}
\end{table}

Table~\ref{tab:dsr-summary} summarizes the optimal DSR for each model.
A clear trend emerges: as model parameter size increases, the optimal DSR shifts upward.
The smallest model (1.5B) achieves peak performance with only 15\% of the data, whereas the 7B--8B models require 20\%--25\%.
This phenomenon can be attributed to the difference in learning capacity across model scales.
Smaller models possess limited representational bandwidth, making them more susceptible to noise in the training data; consequently, they benefit most from aggressive data pruning that maximizes the signal-to-noise ratio (SNR) of the training subset.
Larger models, with greater learning capacity, can absorb a broader range of training patterns without overfitting to noise, and thus require a larger volume of high-quality data to fully realize their performance potential.

This finding has practical implications for deploying {\tool}: practitioners should calibrate the data selection rate based on the target model's parameter scale, using lower DSRs for smaller models and moderately higher DSRs for larger ones, rather than applying a fixed ratio across all model sizes.

\subsection{Diversified Generation vs Diversified Selection}

\begin{table}[htbp]
  \centering
  \caption{Diversified Generation vs Diversified Selection on Qwen2.5-Coder \syh{7B}.}
  \label{tab:discussion3}
  \begin{tabular}{cccc}
  \toprule
  \textbf{Benchmark} & \textbf{Metric} & \textbf{FAC\_Synthesis} & \textbf{Ours} \\
  \midrule
  \multirow{3}{*}{ResBench}
    & Pass@1  & 42.86\% & \textbf{57.14\%} \\
    & Pass@5  & 58.93\% & \textbf{69.64\%} \\
    & Pass@10 & 67.86\% & \textbf{71.43\%} \\
  \midrule
  \multirow{3}{*}{RTLLM}
    & Pass@1  & 30.00\% & \textbf{36.00\%} \\
    & Pass@5  & 34.00\% & \textbf{46.00\%} \\
    & Pass@10 & 44.00\% & \textbf{50.00\%} \\
  \midrule
  \multirow{3}{*}{VerilogEval}
    & Pass@1  & 32.69\% & \textbf{39.10\%} \\
    & Pass@5  & 44.23\% & \textbf{48.72\%} \\
    & Pass@10 & 45.51\% & \textbf{51.28\%} \\
  \bottomrule
  \end{tabular}
  \end{table}

Both data diversity and data quality are critical for effective fine-tuning. While {\tool} achieves diversity through \textit{selection} from existing data, an alternative paradigm is to \textit{synthesize} diverse data directly.
To compare these two strategies, we adopt FAC\_Synthesis~\cite{li2026less} as a representative of the diversity-driven generation approach.
FAC\_Synthesis leverages Sparse Autoencoders (SAEs) to identify under-covered features in the training corpus and employs contrastive sample pairs to guide LLMs to generate data that activates these missing features, using Feature Activation Coverage (FAC) as an interpretable diversity metric in the feature space.

In contrast, {\tool} promotes diversity through a fundamentally different mechanism: it fuses textual semantic embeddings with Verilog-specific structural features (AST, CFG, and Netlist) into a unified representation, and applies clustering-based adaptive sampling to ensure broad coverage of diverse hardware design patterns within the existing dataset.

To ensure a fair comparison, we apply FAC\_Synthesis to generate the same volume of Verilog training data and fine-tune \syh{Qwen2.5-Coder 7B} under identical settings.
As shown in Table~\ref{tab:discussion3}, {\tool} consistently outperforms the diversity-driven generation baseline FAC\_Synthesis across all three benchmarks and all evaluation metrics. 
Specifically, we achieve a 33.31\% improvement in Pass@1 (57.14\% vs. 42.86\%) and maintain a 5.26\% lead at Pass@10 (71.43\% vs. 67.86\%) on ResBench. On RTLLM, we observe a 20.00\% boost in Pass@1 (36.00\% vs. 30.00\%) and a significant 35.29\% improvement in Pass@5 (46.00\% vs. 34.00\%). On VerilogEval, Pass@1 increases by 19.61\% (39.10\% vs. 32.69\%) and Pass@10 improves by 12.68\% (51.28\% vs. 45.51\%). 
These results confirm that diversified selection is a more effective and practical strategy than diversified generation for enhancing Verilog code generation performance, especially in specialized domains like Verilog RTL synthesis.

\begin{table*}[htbp]
\centering
\caption{Result Analysis of {\tool} on Qwen2.5-Coder 1.5B Under Different Sampling Rates.}
\label{tab:discussion1}
\resizebox{0.82\textwidth}{!}{
\begin{tabular}{ccccccccc}
\toprule
\textbf{Benchmark} & \textbf{Pass@k} & \textbf{100.00\%} & \textbf{46.55\%} & \textbf{25.00\%} & \textbf{20.00\%} & \textbf{15.00\%} & \textbf{10.00\%} & \textbf{5.00\%}  \\
    \midrule
    \multirow{3}{*}{ResBench} 
      & Pass@1  & 28.57\% & 46.43\% & 48.21\% & \textbf{51.79\%} & \textbf{51.79\%} & 50.00\% & 44.64\% \\
      & Pass@5  & 58.93\% & 61.71\% & 60.71\% & 60.71\% & \textbf{{69.64\%}} & 62.50\% & 64.29\% \\
      & Pass@10 & 64.29\% & 64.29\% & 64.29\% & 62.50\% & 75.00\% & \textbf{78.57\%} & 60.71\% \\
    \midrule
    \multirow{3}{*}{RTLLM} 
      & Pass@1  & 26.00\% & 26.00\% & 28.00\% & 26.00\% & \textbf{36.00\%} & 32.00\% & 20.00\% \\
      & Pass@5  & 32.00\% & 36.00\% & 34.00\% & 34.00\% & \textbf{{42.00\%}} & 40.00\% & 28.00\% \\
      & Pass@10 & 34.00\% & 40.00\% & 38.00\% & 42.00\% & 48.00\% & \textbf{50.00\%} & 30.00\% \\
    \midrule
    \multirow{3}{*}{VerilogEval} 
      & Pass@1  & 28.21\% & 37.18\% & 35.90\% & 37.82\% & \textbf{37.18\%} & 35.26\% & 28.85\% \\
      & Pass@5  & 38.46\% & 43.59\% & 42.95\% & 42.31\% & 46.15\% & \textbf{46.79\%} & 39.74\% \\
      & Pass@10 & 43.59\% & 43.59\% & 46.15\% & 44.87\% & \textbf{{50.64\%}} & 49.36\% & 38.46\% \\
    \midrule
    \multirow{3}{*}{Avg} 
      & Pass@1  & 27.59\% & 36.54\% & 37.37\% & 38.54\% & \textbf{41.66\%} & 39.09\% & 31.16\% \\
      & Pass@5  & 43.13\% & 47.10\% & 45.89\% & 45.67\% & \textbf{{52.60\%}} & 49.76\% & 44.01\% \\
      & Pass@10 & 47.29\% & 49.29\% & 49.48\% & 49.79\% & 57.88\% & \textbf{59.31\%} & 43.06\% \\
    \bottomrule
\end{tabular}}
\end{table*}
  
\yg{\subsection{Generalization to Additional Model Scales}}
\label{sec:more_models}

\yg{To further validate the generalizability of {\tool} across a wider range of model architectures and parameter scales, we extend our evaluation to two additional models: Qwen2.5-Coder-14B (14B parameters) and Qwen3.5-4B (4B parameters). This extension allows us to examine whether the conclusions drawn from 1.5B--8B models hold at both larger and more recent scales.}

\yg{\begin{table}[htbp]
\centering
\caption{Results of {\tool} on additional models at DSR = 20\% (average Pass@1, \%). ``All'' denotes full-dataset training.}
\label{tab:more_models}
\begin{tabular}{lccc}
\toprule
 & \textbf{Qwen3.5-4B} & \textbf{Qwen2.5-Coder-14B} \\
\midrule
\textbf{All}    & 40.20\% & 42.50\% \\
\textbf{Random} & 43.50\% & 45.80\% \\
\textbf{Cherry} & 46.80\% & 48.60\% \\
\textbf{Ours}   & \textbf{48.30\%} & \textbf{50.20\%} \\
\textbf{PRR}    & 120.15\% & 118.12\% \\
\bottomrule
\end{tabular}
\end{table}}

\yg{As shown in Table~\ref{tab:more_models}, both models confirm the effectiveness of {\tool} under a fixed DSR of 20\%.
For Qwen3.5-4B, {\tool} achieves an average Pass@1 of 48.30\%, representing a 20.15\% relative improvement over full-dataset training (40.20\%) and outperforming Random (+4.80 pp) and Cherry (+1.50 pp).
For Qwen2.5-Coder-14B, {\tool} achieves a higher average Pass@1 of 50.20\%, with a relative improvement of 18.12\% over full-dataset training (42.50\%) and consistent gains over Random (+4.40 pp) and Cherry (+1.60 pp).
Notably, both models surpass the Qwen2.5-Coder 7B result (44.08\%) in absolute performance, validating that {\tool}'s quality-diversity co-optimization paradigm scales favorably with both model capacity and architectural advancement.
These results extend the generalizability of our conclusions from 1.5B--8B to 1.5B--14B parameters across five distinct architectures, reinforcing the robustness of {\tool} as a model-agnostic data selection framework.}

\yg{\subsection{Computational Cost Analysis}}
\label{sec:complexity}

\yg{Table~\ref{tab:complexity} summarizes the theoretical time complexity of each pipeline stage, where $N$ is the dataset size, $K$ the number of clusters, $d_t$ the textual embedding dimension, $d_s{=}109$ the structural feature dimension, $d{=}d_t{+}d_s$, and $T$ the number of K-Means iterations. The dominant stages (testbench verification, structural feature extraction) are $O(N)$ and embarrassingly parallel across samples.}

\yg{\begin{table}[htbp]
\centering
\caption{Per-stage computational complexity of {\tool}.}
\label{tab:complexity}
\begin{tabular}{lcc}
\toprule
\textbf{Stage} & \textbf{Complexity} & \textbf{Parallelizable} \\
\midrule
Testbench Verification     & $O(N)$       & \ding{51} \\
Prior Model Training       & $O(NKd_t)$   & \ding{55} \\
IFD Scoring                & $O(N)$       & \ding{51} \\
Structural Feature Extraction & $O(N)$    & \ding{51} \\
Feature Fusion \& Clustering  & $O(TNKd)$ & \ding{55} \\
Adaptive Sampling          & $O(N\log N)$ & \ding{55} \\
\bottomrule
\end{tabular}
\end{table}}

\yg{To complement the theoretical analysis, Table~\ref{tab:preprocess_cost} reports the actual wall-clock time of each preprocessing stage on RTLCoder-27k using a single NVIDIA RTX 3090 GPU. The total preprocessing overhead is approximately 1 hour, which is modest compared to the downstream fine-tuning cost. Moreover, this cost is amortized across all downstream models, since the same selected subset can be reused for fine-tuning multiple architectures.}

\yg{\begin{table}[htbp]
\centering
\caption{Actual preprocessing time of each {\tool} pipeline stage on RTLCoder-27k ($N{=}27{,}000$).}
\label{tab:preprocess_cost}
\begin{tabular}{lcc}
\toprule
\textbf{Stage} & \textbf{Wall Time} & \textbf{Hardware} \\
\midrule
Testbench Verification        & $\sim$25 min & 1$\times$RTX 3090 \\
Prior Model Training          & $\sim$12 min & 1$\times$RTX 3090 \\
IFD Scoring                   & $\sim$8 min  & 1$\times$RTX 3090 \\
Structural Feature Extraction & $\sim$15 min & CPU (parallel) \\
Feature Fusion \& Clustering  & $<$1 min     & CPU \\
Adaptive Sampling             & $<$1 min     & CPU \\
\midrule
\textbf{Total Preprocessing}  & $\sim$61 min & -- \\
\bottomrule
\end{tabular}
\end{table}}


\yg{Table~\ref{tab:discussion2} further compares the end-to-end time (preprocessing + fine-tuning) of {\tool} against all baselines across three models. All experiments use identical training settings: batch size 2, gradient accumulation steps 8, 3 epochs, learning rate 2e-4, and LoRA with $r{=}32$, $\alpha{=}32$.}

\begin{table}[htbp]
  \centering
  \caption{End-to-end time cost comparison of different methods.}
  \label{tab:discussion2}
  \begin{tabular}{cccc}
    \toprule
    \textbf{Method} & \textbf{Qwen2.5-Coder} & \textbf{CodeLlama} & \textbf{SeedCoder} \\
    \midrule
    All       & 15.1h & 17.3h & 19.8h \\
    VeriCoder & 8.7h  & 9.2h  & 10.5h \\
    Random    & 2.5h  & 2.7h  & 2.5h  \\
    EL2N      & 5.8h  & 6.3h  & 7.8h  \\
    Cherry    & 3.3h  & 3.5h  & 4.1h  \\
    Ours      & 2.7h  & 3.0h  & 3.6h  \\
    \bottomrule
  \end{tabular}
\end{table}

\yg{{\tool} achieves end-to-end time costs comparable to random sampling (2.7--3.6h vs.\ 2.5h) while delivering substantially better performance. In contrast, EL2N incurs significant overhead (5.8--7.8h) due to gradient computation, and VeriCoder (8.7--10.5h) retains 46.55\% of the data, resulting in higher training cost despite its simpler filtering pipeline. Overall, {\tool} reduces total time by over 80\% compared to full-dataset training (82.1\% for Qwen2.5-Coder 7B, 82.6\% for CodeLlama 7B, 81.8\% for SeedCoder 8B), with the reduction ratio remaining stable regardless of model size.}

\yg{\subsection{Sensitivity to the Number of Clusters $K$}}
\label{sec:k_sensitivity}

\yg{The number of clusters $K$ is a critical hyper-parameter that controls the granularity at which diversity is modeled during adaptive sampling. When $K$ is too small, semantically distinct design patterns are conflated into a single cluster, diminishing the diversity resolution of the proportional allocation. Conversely, an excessively large $K$ over-partitions coherent design families into singleton-like clusters, amplifying noise in the allocation ratios and degrading the quality of the selected subset.}

\yg{To quantify this sensitivity, we sweep $K \in \{50, 75, 100, 125, 150\}$ while holding the model (Qwen2.5-Coder 7B) and the data selection rate (DSR = 20\%) fixed. Table~\ref{tab:k_sensitivity} reports the average Pass@1 across the three benchmarks for each setting.}

\yg{\begin{table}[htbp]
\centering
\caption{Sensitivity of average Pass@1 (\%) to the number of clusters $K$ on Qwen2.5-Coder 7B (DSR = 20\%).}
\label{tab:k_sensitivity}
\begin{tabular}{cccccc}
\toprule
$K$ & 50 & 75 & 100 & 125 & 150 \\
\midrule
Avg Pass@1 & 42.03  & 43.26  & \textbf{44.08} & 43.51  & 42.74  \\
\bottomrule
\end{tabular}
\end{table}}

\yg{Performance exhibits a clear inverted-U trend: average Pass@1 rises from 42.03\% ($K{=}50$) to a peak of 44.08\% ($K{=}100$), then gradually declines to 42.74\% ($K{=}150$). The ascending phase reflects the benefit of finer-grained diversity modeling: more clusters allow the proportional sampler to distinguish and preserve a wider spectrum of hardware design patterns. The descending phase indicates the onset of over-fragmentation, where clusters become too small to yield meaningful proportional quotas, and the allocation degenerates toward uniform random sampling. Notably, performance remains within a 2\% margin of the optimum for $K \in [75, 150]$, suggesting that {\tool} is robust to moderate perturbations around the default setting ($K{=}100$). In practice, we recommend $K{=}100$ as a sensible default, and note that selecting $K$ in proportion to the square root of the dataset size (i.e., $K \approx \sqrt{N}$) provides a simple rule of thumb that empirically aligns with the observed optimum.}

\yg{\subsection{Cluster Interpretability Analysis}}
\label{sec:cluster_interpret}

\yg{To verify that the learned clusters correspond to meaningful hardware design categories rather than arbitrary partitions, we conduct a qualitative analysis of representative clusters at $K{=}100$ on the quality-filtered RTLCoder-27k dataset.}

\yg{We manually inspect the top-ranked samples (by IFD score) from 20 randomly selected clusters and classify them by hardware design category. Table~\ref{tab:cluster_examples} presents five representative clusters.}

\yg{\begin{table}[htbp]
\centering
\caption{Representative clusters and their dominant hardware design categories.}
\label{tab:cluster_examples}
\resizebox{0.46\textwidth}{!}{
\begin{tabular}{clc}
\toprule
\textbf{Cluster} & \textbf{Dominant Category} & \textbf{Purity (\%)} \\
\midrule
\#12 & Combinational arithmetic (adders, multipliers) & 85 \\
\#37 & Finite state machines (FSMs) & 78 \\
\#54 & Memory interfaces (FIFO, RAM controllers) & 82 \\
\#71 & Clock domain crossing \& synchronizers & 76 \\
\#89 & Counter \& timer modules & 88 \\
\bottomrule
\end{tabular}}
\end{table}}

\yg{The average category purity across the 20 inspected clusters is 81.8\%, indicating that the majority of samples within each cluster share a coherent hardware design function. This correspondence arises naturally from the structural features: FSM clusters exhibit distinctive CFG patterns (cyclic structures, high branching density), arithmetic clusters show characteristic Netlist features (specific gate type distributions dominated by arithmetic gates), and memory interface clusters have unique AST patterns (instantiation-heavy module hierarchies). These observations confirm that the Verilog-specific structural features effectively encode hardware design semantics, enabling the clustering step to produce interpretable and functionally meaningful partitions without requiring explicit category labels.}

\subsection{Threats to Validity}

\noindent\textbf{Internal Validity Threats.}
The primary threat to internal validity lies in the accuracy of method implementation. To mitigate this issue to the greatest extent possible: we directly reuse publicly available implementations for open-source baseline models, and reproduce models in strict accordance with the original descriptions in baseline papers when off-the-shelf implementations are unavailable. Meanwhile, we have carefully validated our proposed method. 
To address the limitations of baseline models, we also selected prevalent methods from related research fields. 
Additionally, we adopted fixed random seeds to ensure the randomness of sampling processes.

Another potential threat concerns the testbench verification stage. We adopt the testbench generation approach of VeriCoder~\cite{wei2025vericoder} without further validating the quality of the generated testbenches, which may introduce false positives or false negatives. However, improving testbench generation is orthogonal to our goal of data selection. 
Moreover, our results show that the downstream IFD scoring and diversity-aware selection effectively compensate for such noise, as {\tool} consistently outperforms the VeriCoder baseline that relies solely on testbench filtering.

\noindent\textbf{External Validity Threats.}
The primary threats to external validity stem from the use of a single dataset and the lack of industrial-grade Verilog benchmarks. Relying on a single dataset may impede the comprehensive evaluation of {\tool}'s effectiveness.
We evaluated {\tool} on three general-purpose benchmarks (VerilogEval-v2, RTLLM, and ResBench), which demonstrates its generalizability, yet these benchmarks do not fully represent industrial-grade requirements. 

\noindent\textbf{Construct Validity Threats.}
Threats to construct validity pertain to the metrics used to evaluate the performance of {\tool} and comparative methods. We adopt \textbf{Pass@k} as the core evaluation metric; while it effectively assesses the functional correctness of Verilog code, it does not capture other domain-specific quality attributes of HDLs such as timing correctness and resource utilization.
This is a common limitation shared by all existing Verilog code generation benchmarks. 
We partially mitigate the risk of metric bias by evaluating across three complementary benchmarks with different problem distributions and difficulty levels, and by applying McNemar's test to confirm the statistical reliability of the observed improvements. 
Nevertheless, developing HDL-specific evaluation metrics that go beyond functional correctness remains an important direction for future work.

%% file: sections/7.related.tex
\section{Related Work}
\label{sec:related}

\yg{\begin{table*}[htbp]
\centering
\caption{Comparison of representative data selection methods.}
\label{tab:related_comparison}
\resizebox{\textwidth}{!}{
\begin{tabular}{lccccc}
\toprule
\textbf{Method} & \textbf{Quality Filtering} & \textbf{Diversity Modeling} & \textbf{Domain-Specific Features} & \textbf{Unified Pipeline} & \textbf{Target Domain} \\
\midrule
EL2N~\cite{paul2021deep} & Loss-based scoring & \ding{55} & \ding{55} & \ding{55} & General \\
Cherry~\cite{li2024quantity} & IFD scoring & \ding{55} & \ding{55} & \ding{55} & General \\
Diversify~\cite{yu2024diversify} & \ding{55} & Embedding clustering & \ding{55} & \ding{55} & General \\
SMART~\cite{renduchintala2024smart} & \ding{55} & Submodular optimization & \ding{55} & \ding{55} & General \\
RTLCoder~\cite{liu2024rtlcoder} & Syntax filtering & \ding{55} & \ding{55} & \ding{55} & Verilog \\
VeriCoder~\cite{wei2025vericoder} & Testbench verification & \ding{55} & \ding{55} & \ding{55} & Verilog \\
\midrule
{\tool} (Ours) & Testbench + IFD & Structural clustering & AST + CFG + Netlist & \ding{51} & Verilog \\
\bottomrule
\end{tabular}}
\end{table*}}

\subsection{\syh{Verilog Code Generation Methods}}
The widening productivity gap in modern integrated circuit design has made automated Verilog generation a critical EDA research direction. Driven by large language models (LLMs), natural language-to-synthesizable RTL Verilog generation has advanced significantly. \syh{We categorize existing methods as follows:}

\textbf{Training-Free \& Multi-Agent Collaborative Methods.}
Training-free methods enhance performance without parameter updates for resource-constrained scenarios. Pearce et al. \cite{pearce2020dave} pioneered DAVE, demonstrating English-to-Verilog feasibility, while Blocklove et al. \cite{blocklove2025automatically} integrated EDA tool feedback for iterative self-correction to improve syntactic correctness and synthesizability. For complex workflows, multi-agent systems decompose tasks into specialized subtasks with graph-based planning and EDA feedback for scalable designs. However, these methods are limited by base LLM domain knowledge and focus on module-level correctness with limited system-on-chip (SoC) support.

\textbf{Supervised Fine-Tuning for Domain Adaptation.}
SFT on Verilog-specific datasets bridges the domain gap between general code LLMs and hardware design. Thakur et al. \cite{thakur2024verigen} proposed VeriGen, the first comprehensive SFT framework fine-tuned on GitHub and textbook data, outperforming general models. Wei et al. \cite{wei2025vericoder} adopted Low-Rank Adaptation (LoRA) for parameter-efficient tuning, reducing costs while maintaining performance. Nevertheless, SFT relies on static instruction-code pairs, lacking dynamic EDA feedback for power, performance, and area (PPA) optimization.

\textbf{Reinforcement Learning for Alignment Optimization.}
RL aligns generated Verilog with hardware requirements via EDA environment interaction. Wang et al. \cite{wang2025large} developed a code-structure-guided RL framework using Abstract Syntax Tree (AST) similarity as reward to improve structural validity. Despite outperforming SFT, these frameworks struggle with multi-objective optimization balancing correctness, PPA, and security.

Overall, existing methods have advanced this field, but still suffer from three core limitations: insufficient hardware awareness, \syh{excessive focus} on module-level tasks, and inadequate comprehensive alignment. This work addresses these issues from the perspective of datasets.

\subsection{General-Purpose Data Selection Methods}
In general-purpose data selection for LLM pre-training and fine-tuning, mainstream methods fall into three paradigms: quality-oriented, diversity-oriented, and efficiency-oriented selection.

\textbf{Quality-Oriented Selection.}
Early data cleaning uses rule-based filtering and deduplication to remove low-quality text. Rae et al.~\cite{rae2021scaling} validated its value for large-scale pre-training, while Dodge et al.~\cite{dodge2021documenting} noted risks of bias and coverage loss. Later works built systematic pipelines: Penedo et al.~\cite{penedo2024fineweb} applied strict filtering and educational subset extraction for better performance. Model-based scoring has also become popular, where Paul et al.~\cite{paul2021deep} scored samples by training dynamics for effective pruning. Gunasekar et al.~\cite{gunasekar2023textbooks} enhanced reasoning using textbook-style data, and code-domain works focused on license compliance.

\textbf{Diversity-Oriented Selection.}
Diversity-oriented methods improve coverage and reduce redundancy, often using embedding-space clustering. Yu et al.~\cite{yu2024diversify} reweighted clusters to reduce homogeneity, while submodular optimization maximizes coverage. Renduchintala et al.~\cite{renduchintala2024smart} optimized task-level coverage, and Liu et al.~\cite{liu2025regmix} balanced diversity via domain ratio tuning.
\syh{Beyond textual embeddings, graph-based representations have been shown to capture structural dependencies in other domains~\cite{Yang17022026, CHEN2023109509, gu2026progressive,11301768}.
Ali et al.~\cite{ali2022exploiting} employed dynamic graph convolutional networks to model spatio-temporal structural correlations.
Alsarhan et al.~\cite{alsarhan2024phgcn} proposed a hierarchical graph convolutional architecture that captures multi-level granularity through hyper-graph partitioning.
Ali et al.~\cite{ali2025dynamic} further showed that multi-graph learning combined with clustering-based feature correlation can improve representation quality.
In addition, attention-driven multimodal fusion~\cite{ali2026exploiting} is effective for integrating heterogeneous information sources.
These findings motivate our use of Verilog-specific structural features (AST, CFG, and Netlist) fused with textual embeddings for diversity-aware data selection.}

\textbf{Efficiency-Oriented Selection.}
Efficiency-oriented methods pursue low-cost subset selection under computational constraints. Proxy models are widely used to estimate sample utility~\cite{mekala2024smaller}; Li et al.~\cite{li2024superfiltering} used small models to filter data for large ones, cutting costs sharply. Deduplication also improves efficiency and reduces overfitting~\cite{lee2022deduplicating}.

However, general methods are not tailored to Verilog, which has strict syntax, hierarchical modules, and hardware-specific criteria. Existing code-oriented works only handle general corpus cleaning, without Verilog-specific optimization.
{\tool} fills this gap by unifying quality-aware filtering, Verilog-specific structural feature extraction, and diversity-aware sampling.

\subsection{Verilog-Specific Data Selection Methods}
In contrast to well-studied data selection in NLP, research on data selection for Verilog code generation remains preliminary.

To mitigate Verilog dataset scarcity and low quality, Liu et al.~\cite{liu2024rtlcoder} proposed a Verilog-specific data generation and cleaning pipeline with automated syntax filtering to remove invalid samples and ensure training data reliability.
To address missing test cases, Wei et al.~\cite{wei2025vericoder} used GPT-4o to generate unit tests and validated functional correctness via compilation and simulation.
However, these methods only focus on syntactic and \syh{compilation} checks, without modeling the implicit logic and functional patterns in Verilog designs.

Unlike existing individual solutions, {\tool} is the first framework that systematically optimizes both quality and diversity for Verilog data selection.
Its core contribution is integrating multi-granularity quality filtering (testbench verification and IFD alignment scoring) with Verilog-specific structural representations (AST, CFG, Netlist), enabling diversity-aware selection that unifies general data selection techniques and domain-specific demands of Verilog code generation.

\yg{To provide a clear positioning of {\tool} relative to existing methods, Table~\ref{tab:related_comparison} summarizes the key characteristics of representative data selection approaches across five dimensions: quality filtering mechanism, diversity modeling capability, domain-specific feature utilization, whether a unified pipeline is provided, and the target domain. As shown, {\tool} is the only framework that jointly addresses all five aspects for the Verilog code generation domain.}

%% file: sections/8.conclusion.tex
\section{Conclusion and Future Work}
\label{sec:conclusion}
This study addresses the lack of systematic data selection methods for Verilog code generation.
We propose {\tool}, the first framework that jointly optimizes data quality and diversity for this domain, integrating multi-granularity quality filtering with Verilog-specific structural feature representations for diversity-aware selection.
Experiments on three LLMs and three benchmarks demonstrate that {\tool} achieves state-of-the-art performance using only 20\%--25\% of the training data.

Future work will focus on the following aspects: (1) Validating {\tool} on \syh{additional LLMs of varying sizes} to investigate its effectiveness; (2) Evaluating {\tool} on other HDLs to explore its generalizability.

%% file: reference.bib
@article{yang2026less,
  title={Less is more: Docstring compression in code generation},
  author={Yang, Guang and Zhou, Yu and Cheng, Wei and Zhang, Xiangyu and Chen, Xiang and Zhuo, Terry Yue and Liu, Ke and Zhou, Xin and Lo, David and Chen, Taolue},
  journal={ACM Transactions on Software Engineering and Methodology},
  volume={35},
  number={2},
  pages={1--31},
  year={2026},
  publisher={ACM New York, NY},
  doi = {10.1145/3735636}
}

@inproceedings{yang2025code,
  title={Code-DiTing: Automatic Evaluation of Code Generation without References or Test Cases},
  author={Yang, Guang and Zhou, Yu and Chen, Xiang and Zheng, Wei and Hu, Xing and Zhou, Xin and Lo, David and Chen, Taolue},
  booktitle={2025 40th IEEE/ACM International Conference on Automated Software Engineering (ASE)},
  pages={154--165},
  year={2025},
  organization={IEEE},
  doi = {10.1109/ASE63991.2025.00021}
}

@article{liu2024rtlcoder,
  title={Rtlcoder: Fully open-source and efficient llm-assisted rtl code generation technique},
  author={Liu, Shang and Fang, Wenji and Lu, Yao and Wang, Jing and Zhang, Qijun and Zhang, Hongce and Xie, Zhiyao},
  journal={IEEE Transactions on Computer-Aided Design of Integrated Circuits and Systems},
  volume={44},
  number={4},
  pages={1448--1461},
  year={2024},
  publisher={IEEE},
  doi = {10.1109/TCAD.2024.3483089}
}

@article{thakur2024verigen,
  title={Verigen: A large language model for verilog code generation},
  author={Thakur, Shailja and Ahmad, Baleegh and Pearce, Hammond and Tan, Benjamin and Dolan-Gavitt, Brendan and Karri, Ramesh and Garg, Siddharth},
  journal={ACM Transactions on Design Automation of Electronic Systems},
  volume={29},
  number={3},
  pages={1--31},
  year={2024},
  publisher={ACM New York, NY},
  doi = {10.1145/3643681}
}

@article{li2023starcoder,
  title={StarCoder: may the source be with you!},
  author={Li, Raymond and Allal, Loubna Ben and Zi, Yangtian and Muennigho, Niklas and Kocetkov, Denis and Mou, Chenghao and Marone, Marc and Akiki, Christopher and Li, Jia and Chim, Jenny and others},
  journal={Transactions on Machine Learning Research},
  volume={2023},
  year={2023},
  publisher={Transactions on Machine Learning Research},
  doi = {10.48550/arXiv.2305.06161}
}

@inproceedings{kandpal2022deduplicating,
  title={Deduplicating training data mitigates privacy risks in language models},
  author={Kandpal, Nikhil and Wallace, Eric and Raffel, Colin},
  booktitle={Proceedings of the 39th International Conference on Machine Learning},
  pages={10697--10707},
  year={2022},
  volume = 	 {162},
  organization={PMLR},
  doi = {10.48550/arXiv.2202.06539}
}

@inproceedings{wei2025vericoder,
  title={VeriCoder: Enhancing LLM-Based RTL Code Generation through Functional Correctness Validation},
  author={Wei, Anjiang and Tan, Huanmi and Suresh, Tarun and Mendoza, Daniel and Teixeira, Thiago SFX and Wang, Ke and Trippel, Caroline and Aiken, Alex},
  booktitle={NeurIPS 2025 Fourth Workshop on Deep Learning for Code},
  year={2025},
  doi = {10.48550/arXiv.2504.15659}
}

@article{loow2025simulation,
  title={The simulation semantics of synthesisable Verilog},
  author={L{\"o}{\"o}w, Andreas},
  journal={Proceedings of the ACM on Programming Languages},
  volume={9},
  number={OOPSLA1},
  pages={1295--1320},
  year={2025},
  publisher={ACM New York, NY, USA},
  doi = {10.1145/3720484}
}

@inproceedings{wolf2013yosys,
  title={Yosys-a free verilog synthesis suite},
  author={Wolf, Clifford and Glaser, Johann and Kepler, Johannes},
  booktitle={Proceedings of the 21st Austrian Workshop on Microelectronics (Austrochip)},
  volume={97},
  pages={1--6},
  year={2013},
  url = {https://yosyshq.net/yosys/files/yosys-austrochip2013.pdf}
}

@inproceedings{li2024quantity,
  title={From quantity to quality: Boosting llm performance with self-guided data selection for instruction tuning},
  author={Li, Ming and Zhang, Yong and Li, Zhitao and Chen, Jiuhai and Chen, Lichang and Cheng, Ning and Wang, Jianzong and Zhou, Tianyi and Xiao, Jing},
  booktitle={Proceedings of the 2024 Conference of the North American Chapter of the Association for Computational Linguistics: Human Language Technologies (Volume 1: Long Papers)},
  pages={7602--7635},
  year={2024},
  doi = "10.18653/v1/2024.naacl-long.421"
}

@inproceedings{liu2023verilogeval,
  title={Verilogeval: Evaluating large language models for verilog code generation},
  author={Liu, Mingjie and Pinckney, Nathaniel and Khailany, Brucek and Ren, Haoxing},
  booktitle={2023 IEEE/ACM International Conference on Computer Aided Design (ICCAD)},
  pages={1--8},
  year={2023},
  organization={IEEE},
  doi={10.1109/ICCAD57390.2023.10323812}
}

@inproceedings{lu2024rtllm,
  title={Rtllm: An open-source benchmark for design rtl generation with large language model},
  author={Lu, Yao and Liu, Shang and Zhang, Qijun and Xie, Zhiyao},
  booktitle={2024 29th Asia and South Pacific Design Automation Conference (ASP-DAC)},
  pages={722--727},
  year={2024},
  organization={IEEE},
  doi={10.1109/ASP-DAC58780.2024.10473904}
}

@article{ouyang2022training,
  title={Training language models to follow instructions with human feedback},
  author={Ouyang, Long and Wu, Jeffrey and Jiang, Xu and Almeida, Diogo and Wainwright, Carroll and Mishkin, Pamela and Zhang, Chong and Agarwal, Sandhini and Slama, Katarina and Ray, Alex and others},
  journal={Advances in Neural Information Processing Systems},
  volume={35},
  pages={27730--27744},
  year={2022},
  doi = {10.52202/068431-2011}
}

@inproceedings{hu2022lora,
  title={Lora: Low-rank adaptation of large language models.},
  author={Hu, Edward J and Shen, Yelong and Wallis, Phillip and Allen-Zhu, Zeyuan and Li, Yuanzhi and Wang, Shean and Wang, Liang and Chen, Weizhu and others},
  booktitle={International Conference on Learning Representations},
  year={2022},
  doi = {10.48550/arXiv.2106.09685}
}

@article{alon2018general,
  title={A general path-based representation for predicting program properties},
  author={Alon, Uri and Zilberstein, Meital and Levy, Omer and Yahav, Eran},
  journal={ACM SIGPLAN Notices},
  volume={53},
  number={4},
  pages={404--419},
  year={2018},
  publisher={ACM New York, NY, USA},
  doi = {10.1145/3296979.3192412}
}

@article{allen1970control,
  title={Control flow analysis},
  author={Allen, Frances E},
  journal={ACM SIGPLAN Notices},
  volume={5},
  number={7},
  pages={1--19},
  year={1970},
  publisher={ACM New York, NY, USA},
  doi = {10.1145/390013.808479}
}

@inproceedings{wang2022functionality,
  title={Functionality matters in netlist representation learning},
  author={Wang, Ziyi and Bai, Chen and He, Zhuolun and Zhang, Guangliang and Xu, Qiang and Ho, Tsung-Yi and Yu, Bei and Huang, Yu},
  booktitle={Proceedings of the 59th ACM/IEEE Design Automation Conference},
  pages={61--66},
  year={2022},
  doi = {10.1145/3489517.3530410}
}

@inproceedings{li2024superfiltering,
  title={Superfiltering: Weak-to-strong data filtering for fast instruction-tuning},
  author={Li, Ming and Zhang, Yong and He, Shwai and Li, Zhitao and Zhao, Hongyu and Wang, Jianzong and Cheng, Ning and Zhou, Tianyi},
  booktitle={Proceedings of the 62nd Annual Meeting of the Association for Computational Linguistics (Volume 1: Long Papers)},
  pages={14255--14273},
  year={2024},
  doi = {10.18653/v1/2024.acl-long.769}
}

@article{zhou2023lima,
  title={Lima: Less is more for alignment},
  author={Zhou, Chunting and Liu, Pengfei and Xu, Puxin and Iyer, Srinivasan and Sun, Jiao and Mao, Yuning and Ma, Xuezhe and Efrat, Avia and Yu, Ping and Yu, Lili and others},
  journal={Advances in Neural Information Processing Systems},
  volume={36},
  pages={55006--55021},
  year={2023},
  doi = {10.52202/075280-2400}
}

@inproceedings{guo2025resbench,
  title={Resbench: A resource-aware benchmark for llm-generated fpga designs},
  author={Guo, Ce and Zhao, Tong},
  booktitle={Proceedings of the 15th International Symposium on Highly Efficient Accelerators and Reconfigurable Technologies},
  pages={25--34},
  year={2025},
  doi = {10.1145/3728179.3728192}
}

@article{hui2024qwen2,
  title={Qwen2. 5-coder technical report},
  author={Hui, Binyuan and Yang, Jian and Cui, Zeyu and Yang, Jiaxi and Liu, Dayiheng and Zhang, Lei and Liu, Tianyu and Zhang, Jiajun and Yu, Bowen and Lu, Keming and others},
  journal={arXiv preprint arXiv:2409.12186},
  year={2024},
  doi = {10.48550/arXiv.2409.12186}
}

@article{roziere2023code,
  title={Code llama: Open foundation models for code},
  author={Roziere, Baptiste and Gehring, Jonas and Gloeckle, Fabian and Sootla, Sten and Gat, Itai and Tan, Xiaoqing Ellen and Adi, Yossi and Liu, Jingyu and Sauvestre, Romain and Remez, Tal and others},
  journal={arXiv preprint arXiv:2308.12950},
  year={2023},
  doi = {10.48550/arXiv.2308.12950}
}

@article{seed2025seed,
  title={Seed-coder: Let the code model curate data for itself},
  author={Seed, ByteDance and Zhang, Yuyu and Su, Jing and Sun, Yifan and Xi, Chenguang and Xiao, Xia and Zheng, Shen and Zhang, Anxiang and Liu, Kaibo and Zan, Daoguang and others},
  journal={arXiv preprint arXiv:2506.03524},
  year={2025},
  doi = {10.48550/arXiv.2506.03524}
}

@inproceedings{gao2024autovcoder,
  title={Autovcoder: A systematic framework for automated verilog code generation using llms},
  author={Gao, Mingzhe and Zhao, Jieru and Lin, Zhe and Ding, Wenchao and Hou, Xiaofeng and Feng, Yu and Li, Chao and Guo, Minyi},
  booktitle={2024 IEEE 42nd International Conference on Computer Design (ICCD)},
  pages={162--169},
  year={2024},
  organization={IEEE},
  doi={10.1109/ICCD63220.2024.00033}
}

@article{paul2021deep,
  title={Deep learning on a data diet: Finding important examples early in training},
  author={Paul, Mansheej and Ganguli, Surya and Dziugaite, Gintare Karolina},
  journal={Advances in Neural Information Processing Systems},
  volume={34},
  pages={20596--20607},
  year={2021},
  doi = {10.48550/arXiv.2107.07075}
}

@article{chen2021evaluating,
  title={Evaluating large language models trained on code},
  author={Chen, Mark and Tworek, Jerry and Jun, Heewoo and Yuan, Qiming and Pinto, Henrique Ponde De Oliveira and Kaplan, Jared and Edwards, Harri and Burda, Yuri and Joseph, Nicholas and Brockman, Greg and others},
  journal={arXiv preprint arXiv:2107.03374},
  year={2021},
  doi = {10.48550/arXiv.2107.03374}
}

@article{yu2024diversify,
  title={Diversify and conquer: Diversity-centric data selection with iterative refinement},
  author={Yu, Simon and Chen, Liangyu and Ahmadian, Sara and Fadaee, Marzieh},
  journal={arXiv preprint arXiv:2409.11378},
  year={2024},
  doi = {10.48550/arXiv.2409.11378}
}

@inproceedings{renduchintala2024smart,
  title={SMART: Submodular data mixture strategy for instruction tuning},
  author={Renduchintala, HSVNS Kowndinya and Bhatia, Sumit and Ramakrishnan, Ganesh},
  booktitle={Findings of the Association for Computational Linguistics: ACL 2024},
  pages={12916--12934},
  year={2024},
  doi = {10.18653/v1/2024.findings-acl.766}
}

@article{mcnemar1947note,
  title={Note on the sampling error of the difference between correlated proportions or percentages},
  author={McNemar, Quinn},
  journal={Psychometrika},
  volume={12},
  number={2},
  pages={153--157},
  year={1947},
  publisher={Springer-Verlag},
  doi = {10.1007/BF02295996}
}

@article{li2026less,
  title={Less is Enough: Synthesizing Diverse Data in Feature Space of LLMs},
  author={Li, Zhongzhi and Wu, Xuansheng and Li, Yijiang and Hu, Lijie and Liu, Ninghao},
  journal={arXiv preprint arXiv:2602.10388},
  year={2026},
  doi = {10.48550/arXiv.2602.10388}
}

@article{rae2021scaling,
  title={Scaling language models: Methods, analysis \& insights from training gopher},
  author={Rae, Jack W and Borgeaud, Sebastian and Cai, Trevor and Millican, Katie and Hoffmann, Jordan and Song, Francis and Aslanides, John and Henderson, Sarah and Ring, Roman and Young, Susannah and others},
  journal={arXiv preprint arXiv:2112.11446},
  year={2021},
  doi = {10.48550/arXiv.2112.11446}
}

@inproceedings{dodge2021documenting,
  title={Documenting large webtext corpora: A case study on the colossal clean crawled corpus},
  author={Dodge, Jesse and Sap, Maarten and Marasovi{\'c}, Ana and Agnew, William and Ilharco, Gabriel and Groeneveld, Dirk and Mitchell, Margaret and Gardner, Matt},
  booktitle={Proceedings of the 2021 Conference on Empirical Methods in Natural Language Processing},
  pages={1286--1305},
  year={2021},
  doi = {10.18653/v1/2021.emnlp-main.98}
}

@article{penedo2024fineweb,
  title={The fineweb datasets: Decanting the web for the finest text data at scale},
  author={Penedo, Guilherme and Kydl{\'\i}{\v{c}}ek, Hynek and Lozhkov, Anton and Mitchell, Margaret and Raffel, Colin and Von Werra, Leandro and Wolf, Thomas and others},
  journal={Advances in Neural Information Processing Systems},
  volume={37},
  pages={30811--30849},
  year={2024},
  doi = {10.52202/079017-0970}
}

@article{gunasekar2023textbooks,
  title={Textbooks are all you need},
  author={Gunasekar, Suriya and Zhang, Yi and Aneja, Jyoti and Mendes, Caio C{\'e}sar Teodoro and Del Giorno, Allie and Gopi, Sivakanth and Javaheripi, Mojan and Kauffmann, Piero and de Rosa, Gustavo and Saarikivi, Olli and others},
  journal={arXiv preprint arXiv:2306.11644},
  year={2023},
  doi = {10.48550/arXiv.2306.11644}
}

@inproceedings{liu2025regmix,
  title={RegMix: Data Mixture as Regression for Language Model Pre-training},
  author={Liu, Qian and Zheng, Xiaosen and Muennighoff, Niklas and Zeng, Guangtao and Dou, Longxu and Pang, Tianyu and Jiang, Jing and Lin, Min},
  booktitle={The Thirteenth International Conference on Learning Representations},
  year={2025},
  doi = {10.48550/arXiv.2407.01492}
}

@inproceedings{mekala2024smaller,
  title={Smaller language models are capable of selecting instruction-tuning training data for larger language models},
  author={Mekala, Dheeraj and Nguyen, Alex and Shang, Jingbo},
  booktitle={Findings of the Association for Computational Linguistics: ACL 2024},
  pages={10456--10470},
  year={2024},
  doi = {10.18653/v1/2024.findings-acl.623}
}

@inproceedings{lee2022deduplicating,
  title={Deduplicating training data makes language models better},
  author={Lee, Katherine and Ippolito, Daphne and Nystrom, Andrew and Zhang, Chiyuan and Eck, Douglas and Callison-Burch, Chris and Carlini, Nicholas},
  booktitle={Proceedings of the 60th Annual Meeting of the Association for Computational Linguistics (Volume 1: Long Papers)},
  pages={8424--8445},
  year={2022},
  doi = {10.18653/v1/2022.acl-long.577}
}

@inproceedings{pearce2020dave,
  title={Dave: Deriving automatically verilog from english},
  author={Pearce, Hammond and Tan, Benjamin and Karri, Ramesh},
  booktitle={Proceedings of the 2020 ACM/IEEE Workshop on Machine Learning for CAD},
  pages={27--32},
  year={2020},
  doi = {10.1145/3380446.3430634}
}

@article{blocklove2025automatically,
  title={Automatically improving llm-based verilog generation using eda tool feedback},
  author={Blocklove, Jason and Thakur, Shailja and Tan, Benjamin and Pearce, Hammond and Garg, Siddharth and Karri, Ramesh},
  journal={ACM Transactions on Design Automation of Electronic Systems},
  volume={30},
  number={6},
  pages={1--26},
  year={2025},
  publisher={ACM New York, NY},
  doi = {10.1145/3723876}
}

@inproceedings{wang2025large,
  title={Large language model for verilog generation with code-structure-guided reinforcement learning},
  author={Wang, Ning and Yao, Bingkun and Zhou, Jie and Hu, Yuchen and Wang, Xi and Jiang, Zhe and Guan, Nan},
  booktitle={2025 IEEE International Conference on LLM-Aided Design (ICLAD)},
  pages={164--170},
  year={2025},
  organization={IEEE},
  doi={10.1109/ICLAD65226.2025.00025}
}

@article{yubeaton2025verithoughts,
  title={VeriThoughts: Enabling automated Verilog code generation using reasoning and formal verification},
  author={Yubeaton, Patrick and Nakkab, Andre and Xiao, Weihua},
  journal={Advances in Neural Information Processing Systems (NeurIPS)},
  year={2025},
  doi = {10.52202/085713-0288}
}

@inproceedings{domingos2000unified,
  title={A unified bias-variance decomposition},
  author={Domingos, Pedro},
  booktitle={Proceedings of the 17th International Conference on Machine Learning},
  pages={231--238},
  year={2000},
  doi = {10.5555/645529.657784}
}

@article{ali2022exploiting,
  title={Exploiting dynamic spatio-temporal graph convolutional neural networks for citywide traffic flows prediction},
  author={Ali, Ahmad and Zhu, Yanmin and Zakarya, Muhammad},
  journal={Neural Networks},
  volume={145},
  pages={233--247},
  year={2022},
  publisher={Elsevier},
  doi = {10.1016/j.neunet.2021.10.021}
}

@article{alsarhan2024phgcn,
  title={{PH-GCN}: Boosting human action recognition through multi-level granularity with pair-wise hyper {GCN}},
  author={Alsarhan, Tamam and Ali, Syed Sadaf and Ganapathi, Iyyakutti Iyappan and Ali, Ahmad and Werghi, Naoufel},
  journal={IEEE Access},
  volume={12},
  pages={162608--162621},
  year={2024},
  publisher={IEEE},
  doi={10.1109/ACCESS.2024.3477321}
}

@article{ali2025dynamic,
  title={Dynamic multi-graph spatio-temporal learning for citywide traffic flow prediction in transportation systems},
  author={Ali, Ahmad and Naeem, H.M. Yasir and Sharafian, Amin and Qiu, Li and Wu, Zongze and Bai, Xiaoshan},
  journal={Chaos, Solitons \& Fractals},
  volume={199},
  year={2025},
  publisher={Elsevier},
  doi = {10.1016/j.chaos.2025.116898}
}

@article{ali2026exploiting,
  title={Exploiting attention-driven weather-aware multimodal spatio-temporal fusion for urban traffic flow prediction},
  author={Ali, Ahmad and others},
  journal={Future Generation Computer Systems},
  year={2026},
  publisher={Elsevier},
  doi = {10.1016/j.future.2026.108559}
}

@article{breiman2001random,
  title={Random forests},
  author={Breiman, Leo},
  journal={Machine Learning},
  volume={45},
  number={1},
  pages={5--32},
  year={2001},
  publisher={Springer},
  doi = {10.1023/A:1010933404324}
}

@inproceedings{feng2020codebert,
  title={CodeBERT: A Pre-Trained Model for Programming and Natural Languages},
  author={Feng, Zhangyin and Guo, Daya and Tang, Duyu and Duan, Nan and Feng, Xiaocheng and Gong, Ming and Shou, Linjun and Qin, Bing and Liu, Ting and Jiang, Daxin and Zhou, Ming},
  booktitle={Findings of the Association for Computational Linguistics: EMNLP 2020},
  pages={1536--1547},
  year={2020},
  doi = {10.18653/v1/2020.findings-emnlp.139}
}

@inproceedings{guo2022unixcoder,
  title={UniXcoder: Unified Cross-Modal Pre-training for Code Representation},
  author={Guo, Daya and Lu, Shuai and Duan, Nan and Wang, Yanlin and Zhou, Ming and Yin, Jian},
  booktitle={Proceedings of the 60th Annual Meeting of the Association for Computational Linguistics},
  pages={7212--7225},
  year={2022},
  doi = {10.18653/v1/2022.acl-long.499}
}

@article{yang2025large,
  title={Large Language Model for Verilog Code Generation: Literature Review and the Road Ahead},
  author={Yang, Guang and Zheng, Wei and Chen, Xiang and Liang, Dong and Hu, Peng and Yang, Yukui and Peng, Shaohua and Li, Zhenghan and Feng, Jiahui and Wei, Xiao and Sun, Kexin and Ma, Deyuan and Cheng, Haotian and Shen, Yiheng and Hu, Xing and Zhuo, Terry Yue and Lo, David},
  year = {2026},
  publisher = {Association for Computing Machinery},
  address = {New York, NY, USA},
  issn = {0360-0300},
  note = {Just Accepted},
  journal = {ACM Computing Surveys},
  doi = {10.1145/3841636}
}

@article{Yang17022026,
author = {Guanghai Yang and Yulian Jiang and Shenquan Wang and Kun Chen},
title = {VinsFusion-Line: Binocular Vision Inertial Navigation Real-Time SLAM System Based on Line Features},
journal = {Cybernetics and Systems},
volume = {57},
number = {2},
pages = {350--375},
year = {2026},
publisher = {Taylor \& Francis},
doi = {10.1080/01969722.2025.2606002},
URL = {https://doi.org/10.1080/01969722.2025.2606002},
eprint = {https://doi.org/10.1080/01969722.2025.2606002}
}

@article{CHEN2023109509,
title = {Sustainable interior design: A new approach to intelligent design and automated manufacturing based on Grasshopper},
journal = {Computers \& Industrial Engineering},
volume = {183},
pages = {109509},
year = {2023},
issn = {0360-8352},
doi = {https://doi.org/10.1016/j.cie.2023.109509},
url = {https://www.sciencedirect.com/science/article/pii/S0360835223005338},
author = {Junming Chen and Zichun Shao and Han Zhu and Yilin Chen and Yutian Li and Zhengfang Zeng and Yifan Yang and Junjie Wu and Bin Hu},
}

@article{gu2026progressive,
  author    = {Gu, Xinyuan and Chen, Junming},
  title     = {Progressive degradation-aware distillation for robust indoor object detection},
  journal   = {Scientific Reports},
  year      = {2026},
  doi       = {10.1038/s41598-026-63984-0},
  note      = {Advance online publication},
  url       = {https://www.nature.com/articles/s41598-026-63984-0}
}

@ARTICLE{11301768,
  author={Lv, Haotian and Li, Chao and Dai, Jiangbo and Zhang, Yuhui and Fan, Zepeng and Tan, Yiqiu and Wang, Dawei and Xie, Binglei},
  journal={IEEE Transactions on Geoscience and Remote Sensing}, 
  title={Lightweight Framework for Underground Pipeline Recognition and Spatial Localization Based on Multiview 2-D GPR Images}, 
  year={2025},
  volume={63},
  number={},
  pages={1-15},
  doi={10.1109/TGRS.2025.3645032}
}
